\documentclass[11pt]{article}
\pdfoutput=1
\usepackage{microtype}

\usepackage[colorlinks=true,urlcolor=Blue,citecolor=Green,linkcolor=BrickRed]{hyperref}
\usepackage[usenames,dvipsnames]{xcolor}

\usepackage[margin=1in]{geometry}
\usepackage[utf8]{inputenc}
\usepackage{algorithm}
\usepackage{algorithmicx}
\usepackage[noend]{algpseudocode}
\usepackage[noadjust]{cite}
\usepackage[OT4]{fontenc}
\usepackage{amsthm}
\usepackage{amssymb}
\usepackage{amsmath}
\usepackage{amsfonts}
\usepackage{todonotes}
\usepackage{authblk}
\usepackage{rotating}
\usepackage{xspace}
\usepackage[inline]{enumitem}

\newcommand{\Oh}{\mathcal{O}}
\newcommand{\Ohtilda}{\tilde{\Oh}}
\newcommand{\eps}{\varepsilon}
\newcommand{\QD}{\texttt{QD}(T_1,T_2)}

\newcommand{\mult}{\textsc{mult}}
\newcommand{\pluseq}{\mathrel{+}=}
\newcommand{\TT}{\mathcal{T}}
\newcommand{\pow}{\gamma}
\newcommand{\spec}{\textsc{spec}}

\SetLabelAlign{Left}{\makebox[1.5em][l]{#1}}

\newtheorem{theorem}{Theorem}[section]
\newtheorem{property}[theorem]{Property}
\newtheorem{corollary}[theorem]{Corollary}
\newtheorem{lemma}[theorem]{Lemma}
\newtheorem{observation}[theorem]{Observation}
\newtheorem{fact}[theorem]{Fact}
\newtheorem{conjecture}{Conjecture}
\newtheorem{proposition}[theorem]{Proposition}
\theoremstyle{definition}   
\newtheorem{definition}[theorem]{Definition}

\newtheorem{question}{Question}

\title{Computing Quartet Distance Is Equivalent to Counting 4-Cycles}
\date{}	
\author[1]{Bartłomiej Dudek
}
\author[1]{Paweł Gawrychowski}

\affil[1]{Institute of Computer Science, University of Wrocław, Poland}

\newcommand{\FIGURE}[4]{
\begin{figure}[#1]
\begin{centering}
\includegraphics[width={#2}\textwidth]{{figures/#3}.pdf}
\caption{#4}
\label{fig:#3}
\end{centering}
\end{figure}
}

\begin{document}
\maketitle
 
\thispagestyle{empty}
 
\begin{abstract}
The quartet distance is a measure of similarity used to compare two unrooted
phylogenetic trees on the same set of $n$ leaves, defined as the number of
subsets of four leaves related by a different topology in both trees. After a series
of previous results, Brodal et al. [SODA 2013] presented an algorithm
that computes this number in $\Oh(nd\log n)$ time, where $d$ is the maximum degree of
a node. For the related triplet distance between rooted phylogenetic trees, the same
authors were able to design an $\Oh(n\log n)$ time algorithm, that is, with running time
independent of $d$. This raises the question of achieving such complexity for
computing the quartet distance, or at least improving the dependency on $d$.

Our main contribution is a two-way reduction establishing that the complexity
of computing the quartet distance between two trees on $n$ leaves is the same,
up to polylogarithmic factors, as the complexity of counting 4-cycles in an undirected
simple graph with $m$ edges. The latter problem has been extensively studied, and the fastest
known algorithm by Vassilevska Williams et al. [SODA 2015] works in $\Oh(m^{1.48})$ time.
In fact, even for the seemingly simpler problem of detecting a 4-cycle, the best
known algorithm works in $\Oh(m^{4/3})$ time, and a conjecture of Yuster and Zwick
implies that this might be optimal. In particular, an almost-linear time for
computing the quartet distance would imply a surprisingly efficient algorithm for counting
4-cycles. In the other direction, by plugging in the state-of-the-art algorithms for
counting 4-cycles, our reduction allows us to significantly decrease the complexity
of computing the quartet distance. For trees with unbounded degrees we obtain
an $\Oh(n^{1.48})$ time algorithm, which is a substantial improvement on the previous
bound of $\Oh(n^{2}\log n)$. For trees with degrees bounded by $d$, by analysing
the reduction more carefully, we are able to obtain an $\Ohtilda(nd^{0.77})$\footnote{$\Ohtilda(.)$ hides factors
polylogarithmic in $n$.} time algorithm,
which is again a nontrivial improvement on the previous bound of $\Oh(nd\log n)$.
\end{abstract}
 
\clearpage
\setcounter{page}{1}
 
\section{Introduction}

Many branches of science study evolutionary relationships between objects.
The canonical example is biology with species or gene relationships, but 
similar questions arise also in linguistics looking into related natural
languages~\cite{GrayDG09,WalkerWMA12,NakhlehWRE05}, or archaeology
studying how ancient manuscripts changed over time~\cite{Buneman71}.
In most cases the hierarchical structure is represented as a tree, called a
phylogenetic tree in biological applications. In this paper we focus on
{\it unrooted} phylogenetic trees that describe the relationship between
species mapped to its leaves without making any assumptions about
the ancestry. The main goal is to understand the true relationship between the
objects in question based on often incomplete or noisy data. An additional difficulty
is that the obtained tree depends on the inference method (e.g.
Q*~\cite{BerryG00}, neighbor
joining~\cite{SaitouN87}) and the assumed model.
See~\cite[Chapter 17]{Gusfield97} for an overview of available models and
construction methods.
Consequently, we might be able to infer multiple trees that should be
compared to determine if our results are consistent.

The most common approach for comparing multiple trees is to define a measure of dissimilarity
between two trees. Various metrics have been already defined, e.g. the symmetric difference metric~\cite{RobinsonF79},
the nearest-neighbor interchange metric~\cite{WatermanS78}, the subtree transfer distance \cite{AllenS01},
the Robinson and Foulds distance~\cite{RobinsonF81}, the quartet distance~\cite{EstabrookMM85} and the triplet distance~\cite{Dobson75}.
Each of them has its particular advantages and disadvantages, see the
discussion in~\cite{BandeltD86,SteelP93}, but the quartet-based reconstruction is perhaps the most
studied (see, e.g., \cite{BerryG00,BerryJKLW99, JiangKL98,JohnWMV03, SnirR10, SnirY12,StrimmerH96}).
Most importantly, according to Bryant et al.~\cite{BryantTKL00}, as opposed to some
other methods, it is able to distinguish both between transformations that affect a large number
of leaves and those that affect only a few of them.
The idea is to consider the basic unit of information in such a tree, which is a subtree induced by four
leaves (called a quartet).
See Figure~\ref{fig:four_topologies} for an illustration of the four possible topologies induced by a quartet.
\begin{definition}
Given two trees, each on the same set of leaves, the \textit{quartet distance} is the number of quartets that are related by different topologies in both trees. 
\end{definition}
\noindent
Note that in this context we may assume that there are no internal nodes
of degree 2 in both trees. We assume that the set of leaves corresponds to the full set of species. 

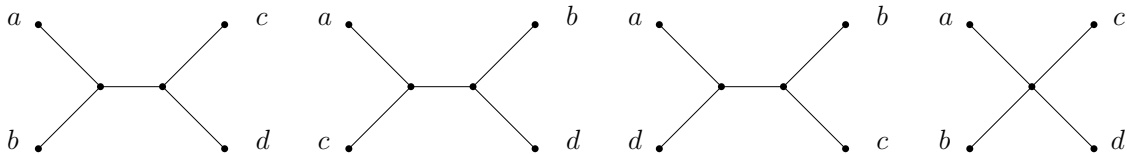
\FIGURE{b}{0.9}{four_topologies}{Four possible topologies of a tree induced by four leaves.
Note that in the induced subtree there might be some internal nodes on each of the edges.}

Quartet distance has been studied from multiple angles. From the combinatorial perspective,
an intriguing question is to investigate the maximum possible quartet distance between two trees on
$n$ leaves. A conjecture of Bandelt and Dress~\cite{BandeltD86} is that this is always $(\frac{2}{3}+o(1)){n\choose 4}$,
with the best known bound being $(0.69+o(1)){n \choose 4}$ by Alon et al.~\cite{AlonNS16}.
From the algorithmic perspective, a long-standing challenge is to compute the quartet distance efficiently.
For trees with all internal nodes of degree 3, a series of papers~\cite{SteelP93,BryantTKL00,BrodalFP01}
has culminated in an $\Oh(n\log n)$ time algorithm by Brodal et al.~\cite{BrodalFP03}.
For the more challenging general case the complexity has been decreased from $\Oh(n^{3})$~\cite{ChristiansenMPR05}
to $\Oh(n^{2.688})$~\cite{NielsenKMP11}
and then, for trees with all internal degrees bounded by $d$,
further to $\Oh(n^{2}d^{2})$~\cite{ChristiansenMPR05}, $\Oh(n^{2}d)$~\cite{ChristiansenMPRS06}, $\Oh(nd^{9}\log n)$~\cite{StissingPMBF07},
and finally to $\Oh(nd\log n)$ by Brodal et al.~\cite{BrodalFMPS13}.
Even though some reconstruction methods produce trees with all internal degrees bounded by 3,
called {\it fully-resolved}, trees that are not fully resolved do appear in some contexts, see e.g.~\cite{Buneman71}
and its refinements. This suggests the following question.

\begin{question}
\label{q:degree}
Can we beat $\Oh(nd\log n)$ for computing the quartet distance between two
trees on $n$ leaves and all internal degrees bounded by $d$?
\end{question}

A related measure is the triplet distance, defined for {\it rooted} phylogenetic trees, where we count
triplets of leaves that are related by the same topology in both trees~\cite{Dobson75}. A successful line of
research~\cite{CritchlowPQ96,BansalDFB11,SandBFPM13} has resulted in an $\Oh(n\log n)$ time algorithm
by Brodal et al.~\cite{BrodalFMPS13} for computing the triplet distance between two arbitrary trees.
The algorithms designed for computing the triplet and quartet distance are based on similar ideas,
see the survey by Sand et al.~\cite{SandHJFBPM13}. Thus it is plausible that, with some additional insight,
we might be able to design an $\Oh(n\log n)$ time algorithm for computing the quartet distance between
two arbitrary trees, without any assumption on their degrees, similarly as for the triplet distance.
Note that the fastest currently known algorithm for the general case works in $\Oh(n^{2}\log n)$ time~\cite{BrodalFMPS13}.
This suggests the following question.

\begin{question}
\label{q:general}
Can we design an $\Oh(n\log n)$ time algorithm for computing the quartet
distance between two trees on $n$ leaves?
\end{question}

\paragraph{Fine grained complexity.}
The traditional notion of ``easy'' and ``hard'' problems is defined with respect to the polynomial time
solvability. However, for many such easy problems the best known algorithms have very high complexities,
making them intractable in practice, in spite of a significant effort from the algorithmic community. This suggests
that the known algorithms are optimal (or at least very close to optimal). Unfortunately, proving unconditional
statements on the optimal complexity doesn't seem within our reach, unless we are willing to work in a severely
restricted model of computation. This spurred a recent systematic effort to create a map of polynomial-time solvable
problems by connecting them to a few believable conjectures on complexities of some basic problems,
such as SETH, APSP, or 3SUM. See a recent survey by Vassilevska Williams~\cite{VirgiICM} for a summary of this effort.

A basic question concerning graphs is to count (or detect) occurrences of certain structures, with perhaps
the most fundamental example being counting triangles, that is, 3-cycles, in a simple undirected graph on $n$ nodes.
Of course, this can be easily solved in $\Oh(n^{\omega})=\Oh(n^{2.38})$ by plugging in the fastest known matrix
multiplication algorithm~\cite{Gall14a,Williams12}. Somewhat surprisingly, Vassilevska Williams and Williams~\cite{Triangle}
proved that these two problems are, in a certain sense, equivalent: a truly subcubic algorithm for detecting triangles
implies a truly subcubic algorithm for Boolean matrix multiplication. For the more practically relevant case of a sparse
undirected graph with $m$ edges, Alon et al.~\cite{AlonYZ97} designed an $\Oh(m^{2\omega/(\omega+1)})=\Oh(m^{1.41})$
time algorithm for counting triangles (their algorithm is stated for finding a triangle, but can be easily extended). Going one
step further, 4-cycles can also be counted in $\Oh(n^{\omega})$ time~\cite{AlonYZ97}, but interestingly one can \textit{find} a $2k$-cycle,
for $k\geq 2$, in $\Oh(n^{2})$ time, as shown by Yuster and Zwick \cite{YusterZ97}.
If the graph is given as an adjacency matrix, this is clearly optimal, but it seems plausible to conjecture that
this is also optimal if the graph is given as adjacency lists.

\begin{conjecture}[Yuster and Zwick \cite{YusterZ97}]\label{conj:no_n2-eps}
 For every $\eps>0$, there is no algorithm that detects 4-cycles in a graph on $n$ nodes in~$\Oh(n^{2-\eps})$ time.
\end{conjecture}

Returning to sparse graphs, Alon et al.~\cite{AlonYZ97} showed how to find a 4-cycle in $\Oh(m^{4/3})$ time, and recently Dahlgaard et al.~\cite{DahlgaardKS17} provided a very nontrivial extension to finding any $2k$-cycle in $\Oh(m^{2k/(k+1)})$ time.
Moreover, they showed that this is optimal, if one is willing to believe Conjecture~\ref{conj:no_n2-eps}, and using a general combinatorial result of Bondy and Simonovits that a graph with $m=100kn^{1+1/k}$ edges must contain a $2k$-cycle \cite{BondyS74}.
See also Abboud and Vassilevska Williams~\cite{AbboudW14} for a similar conjecture on the complexity of detecting a 3-cycle.

\begin{conjecture}[Dahlgaard et al.~\cite{DahlgaardKS17}]\label{conj:no_m4/3-eps}
 For every $\eps>0$, there is no algorithm that detects a 4-cycle in a graph with $m$ edges in~$\Oh(m^{4/3-\eps})$ time.
\end{conjecture}
A related question is to find an occurrence of an induced subgraph. Vassilevska Williams et al.~\cite{WilliamsWWY15}
provide a systematic study of this question for all induced four-node graphs. They also provide an
algorithm that can be used to {\it count} occurrences of a 4-cycle (not necessarily induced) in $\Oh(m^{1.48})$ time.
Also, Abboud et al.~\cite{AbboudWY15} consider a certain generalisation of detecting 3-cycles in which the nodes are colored
and we are asked to check if there exists a 3-cycle for every possible triple of distinct colors.

\paragraph{Our contribution.}
We answer both Question~\ref{q:degree} and Question~\ref{q:general} by connecting the complexity of computing the
quartet distance with the complexity of counting 4-cycles in a simple undirected graph. By providing reductions in both directions we show that
these problems are equivalent, up to polylogarithmic factors. The reductions are summarised in Figure~\ref{fig:plan}.

\begin{figure}[ht!]

\newcommand{\myproblem}[1]{\large\textbf{#1}}
\newcommand{\mydesc}[1]{\small{Section {#1}}}
\newcommand{\myedgeshort}{edge[bend left=10]}
\newcommand{\myedgelong}{edge[bend left=15]}

    \centering
    \begin{tikzpicture}
    \begin{scope}[every node/.style={},
                  every edge/.style={draw=black,very thick}]
                  
	\node (multi) at (0,0) {\myproblem{Multigraphs}};
        \node (multismall) at (3.3,2) {\myproblem{Multigraphs with small multiplicities}};
        \node (simple) at (6,0.75) {\myproblem{Simple graphs}};
        \node (bipartite) at (6,-0.75) {\myproblem{Bipartite simple graphs}};
        \node (quartets) at (3,-2) {\myproblem{Quartet distance}};
        
        \path [->] (multi)      \myedgelong node[left]  	{\mydesc{\ref{se:mult_reduction}}} (multismall);
        \path [->] (multismall) \myedgeshort node[right]	{\mydesc{\ref{se:small_mult}}} (simple);
        \path [->] (simple)  	\myedgeshort node[right]	{\mydesc{\ref{se:simple_to_bipartite}}} (bipartite);
        \path [->] (bipartite)  \myedgeshort node[right]	{\mydesc{\ref{se:counting_4c_to_quartet_d}}} (quartets);
        \path [->] (quartets)   \myedgelong node[left]  	{\mydesc{\ref{se:from_qd_to_4c} and \ref{se:faster_alg_for_qd}}} (multi);

    \end{scope}
    \end{tikzpicture}
    \caption{Summary of the reductions between counting 4-cycles in specific graphs (bipartite or arbitrary, simple or multi)
and computing the quartet distance, all up to polylogarithmic factors.}
    \label{fig:plan}
\end{figure}
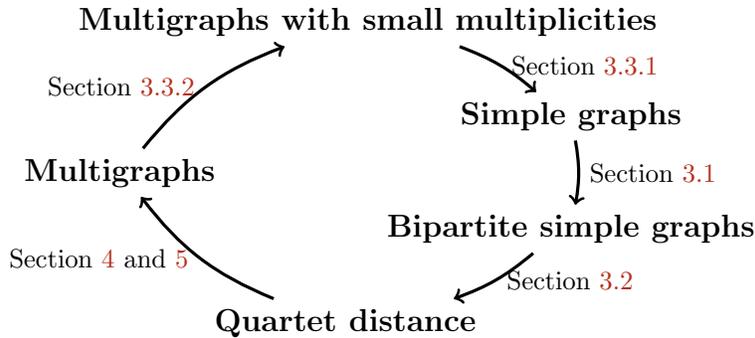

Our reduction from counting 4-cycles in a simple graph to computing the quartet distance implies that an $\Oh(n^{4/3-\eps})$ time algorithm
for computing the quartet distance between two trees on $n$ leaves would imply a surprisingly fast $\Oh(m^{4/3-\eps})$ time
algorithm for counting, and thus also detecting, 4-cycles, thereby refuting Conjecture~\ref{conj:no_m4/3-eps}. 
Note that we create a node of the tree for every edge of the original graph and hence the complexity of the algorithm for detecting 4-cycles implied by our reduction depends on the number of edges, not nodes.
This provides a reasonable explanation of why there has been no $\Oh(n\log n)$ time algorithm for computing the quartet distance.

\begin{proposition}
 There exists no algorithm that can compute the quartet distance between trees on $n$ leaves in $\Oh(n^{4/3-\eps})$ time unless
Conjecture~\ref{conj:no_m4/3-eps} is false.
\end{proposition}

In the other direction, the reduction from computing the quartet distance to multiple instances of counting 4-cycles in a simple graph
allows us to significantly improve on the best known complexity of the former problem by plugging in the state-of-the-art
algorithms for the latter problem. Recall that 4-cycles in a simple graph on $n$ nodes with $m$ edges can be counted in
either $\Oh(n^{2.38})$~\cite{AlonYZ97} or $\Oh(m^{1.48})$ time~\cite{WilliamsWWY15} (the $\Oh(m^{4/3})$ algorithm based on capped $k$-walks \cite{DahlgaardKS17} cannot be applied here, as it merely detects, but does not count the cycles).
Using the latter algorithm we obtain
that the quartet distance between two trees on $n$ leaves can be computed in $\Oh(n^{1.48})$ time, which is a substantial
improvement on the previously known quadratic time bound. Furthermore, for trees with all internal nodes having degrees
bounded by $d$ the running time of the obtained algorithm is $\Ohtilda(nd^{0.96})$,
and if we carefully analyse parameters of the graphs generated by the reduction and switch to the
$\Oh(n^{2.38})$ algorithm for counting 4-cycles in some of them, the complexity further decreases to $\Ohtilda(nd^{0.77})$.

\begin{theorem}\label{thm:new_alg_for_qd}
There exists an algorithm for computing the quartet distance between two trees on $n$ leaves and all internal nodes having degrees
bounded by $d$ in $\Ohtilda(\min\{n^{1.48},nd^{0.77}\})$ time.
\end{theorem}

An important ingredient of our proof is a reduction from counting 4-cycles in a multigraph with edge multiplicities bounded
by $U$ to $\Oh(\log^{4}U)$ instances of counting 4-cycles in simple graphs of roughly the same size. At first this might seem
to be an unnecessary complication, as it is plausible that the $\Oh(m^{1.48})$ algorithm of Vassilevska Williams et al.
\cite{WilliamsWWY15} can be extended, with some effort, to work for multigraphs, and used to show the existence of an $\Oh(n^{1.48})$ time algorithm for computing the quartet distance.
However, it is not completely clear
if every algorithm for this problem can be similarly extended, so further improvements in the complexity
of counting 4-cycles might or might not translate into an improvement for computing the quartet distance.
Furthermore, we don't see how to provide a direct reduction from counting 4-cycles in a multigraph to computing
the quartet distance, so switching to multigraphs wouldn't allow us to state an equivalence between these two problems.
Finally, we believe that our general reduction from multigraphs to simple graphs might be of independent interest.

We note that Jansson and Lingas \cite{JanssonL14} show how to reduce computing triplet distance in so called galled tree to counting triangles in many graphs.
However, this results in $\Oh(n^{2.687})$ algorithm, but in fact an $\Oh(n\log n)$ solution which does not require this idea exists \cite{JanssonRS17}.
We use a significantly different approach that gives us more control on sizes of the obtained subproblems and bound the overall running time.
Furthermore we need to deal with additional technical complications due to the fact that we are working with 4-cycles instead of triangles.

\paragraph{Overview of the methods. }
The $\Oh(nd\log n)$ complexity of the fastest known algorithm for computing the quartet distance suggests
that a difficult instance consists of two trees with high internal degrees, and indeed the trees obtained in our reduction
have small depth but very high degrees.
We start with reducing counting 4-cycles in a simple graph to counting 4-cycles in a simple bipartite graph. This
is easily achieved by duplicating the nodes. Then, we construct two trees of depth 2, each consisting of the root with
its children corresponding to the nodes of the graph. Finally, each edge of the graph corresponds to a leaf attached,
in every tree, to the child of the root corresponding to its appropriate endpoint. The main difficulty in this reduction is
that we need to carefully analyse all possible quartets and bipartite graphs on four edges to argue that, with some additional linear-time computation,
we can extract the number of 4-cycles from the quartet distance.

In the other direction, our reduction is more involved. We first notice that due to the algorithm of Brodal et
al.~\cite{BrodalFMPS13} we only need to show how to efficiently count quartets that are unresolved in both trees, that is,
stars (the rightmost topology in Figure~\ref{fig:four_topologies}). As a first approximation, we could iterate over the potential central nodes of the star in both trees and
create a bipartite multigraph such that counting matchings of size 4 there gives us the number of quartets
with these central nodes. There are at least two issues with this approach. First, we need to prove that counting such
matchings in a multigraph can be reduced to counting 4-cycles in a simple graph. Second, we cannot afford to create
a separate instance for every pair of central nodes, and furthermore even if we were able to decrease their number we
would still need to have some control on the total size of the obtained bipartite graphs.

We overcome the first difficulty in two steps. We begin with reducing counting matchings of size 4 in a multigraph
to counting 4-cycles in a multigraph. This requires a careful analysis of all possible multigraphs on four edges and
extends a similar reasoning used in the other direction of the reduction. Then, we reduce counting 4-cycles in a multigraph
with multiplicities bounded by $U$ to counting 4-cycles in simple graphs of roughly the same size as the original multigraph. This is obtained by first
designing an efficient reduction to a constant number of instances of counting 4-cycles in simple graphs for small (constant) values of $U$.
Then a careful application of polynomial interpolation allows us to obtain $\Oh(\log^4U)$ instances of counting 4-cycles in simple graphs.

The second difficulty is more fundamental. To avoid iterating over all pairs of central nodes, we apply a certain
hierarchical decomposition of both trees known as the top tree decomposition \cite{AlstrupHLT05,BilleGLW15}. A similar decomposition has been
already used by Brodal et al.~\cite{BrodalFMPS13}, but we apply it to both trees simultaneously.
This allows us to decrease the  number of explicitly considered pairs of central nodes to only $\Oh(n\log^{2}n)$
and consider the remaining pairs aggregately in batches. The remaining pairs have a simple structure, but counting
them efficiently requires providing a mechanism for answering certain queries on a tree. This is implemented with
the standard heavy-light decomposition and follows the high-level idea used by Brodal et al.~\cite{BrodalFMPS13}.

\section{Preliminaries}

We consider unrooted trees on $n$ leaves with distinct labels from $\{1,2,\ldots,n\}$, and identify leaves with their labels.
The quartet distance between two such trees $T_1,T_2$ is defined as the number of subsets of four distinct leaves $\{a,b,c,d\}$
(called quartets) such that the subtrees induced by $\{a,b,c,d\}$ in both trees are not related by the same topology.
There are four possible topologies of trees induced by four leaves, see Figure~\ref{fig:four_topologies}.

We work with undirected graphs. Whenever we talk about counting 4-cycles in such a graph we mean simple cycles of length 4,
not necessarily induced. For counting 4-cycles self-loops are irrelevant, but unless stated otherwise there might be multiple edges,
and then we count the cycle multiple times, the product of the multiplicities of its edges.

A multigraph is a triple $(V,E,\mult)$, where $E$ is a set of edges and the function $\mult: E \rightarrow \{1,\ldots,U\}$ returns multiplicities of edges.
Throughout the paper $U$ will be bounded by the total number of edges in the input graph, which sometimes will be much bigger than the size of the currently considered graph.
For simple graphs it holds that $\mult(e)=1$ for all edges $e\in E$ and the function is omitted.

\newcommand{\mysymb}[4]{
\newcommand{#1}{%
  \mathrel{{\ooalign{\hss\raisebox{-#2}{$ #4$}\hss\cr\raisebox{#2}{$ #3$}}}}
}}
\newcommand{\mystack}[7]{
\newcommand{#1}[0]{%
  \mathrel{{\ooalign{
  \hss\cr\raisebox{#2}{$ #3$}
  \hss\cr\raisebox{#4}{$ #5$}
  \hss\cr\raisebox{#6}{$ #7$}}}}\hspace{-0.1cm}
}}

\newcommand{\mirror}[1]{\reflectbox{$ #1$}}
\newcommand{\rotatechar}[1]{%
  \vcenter{\hbox to 1.6ex{\hfil\rotatebox{90}{#1}\hfil}}}
\newcommand{\rotatehor}[1]{%
  \vcenter{\hbox to 1.6ex{\hfil\rotatebox{180}{#1}\hfil}}}
\newcommand{\rotatediag}[1]{%
  \vcenter{\hbox to 1.6ex{\hfil\rotatebox{-45}{#1}\hfil}}}
  
\mysymb{\sG}{0.7ex}{<}{>}
\mysymb{\sI}{0.7ex}{<}{=}
\mysymb{\sJ}{0.5ex}{=}{=}
\mysymb{\sE}{0.63ex}{<}{<}
\mysymb{\sF}{0.9ex}{<}{<}
\mysymb{\sB}{0.61ex}{\angle}{>}

\mystack{\sC}{0.5ex}{\angle}{-0.4ex}{\smallsetminus}{-1.6ex}{-}
\mystack{\sH}{0.5ex}{\angle}{-0.4ex}{\mirror\smallsetminus}{-1.6ex}{-}

\mystack{\sD}{-0.3ex}{\angle}{-0.3ex}{\mirror\angle}{0.67ex}{-}

\mystack{\sA}{0.7ex}{\angle}{-0.3ex}{\smallsetminus}{-0.92ex}{\setminus}

\newcommand{\zeroperp}{\rotatechar{\scriptsize{\textbar}\hspace{-0.055cm}\scriptsize{)}}}
\newcommand{\zerodiag}{\hspace{-0.035cm}\rotatediag{\footnotesize{(}\hspace{-0.057cm}\footnotesize{\textbar}}}

\mystack{\zoz}{0.72ex}{-}{-0.72ex}{-}{0.12ex}{\zerodiag}
\mysymb{\sOII}{0.65ex}{\zeroperp}{=}
\mysymb{\sOV}{0.75ex}{\zeroperp}{<}
\mysymb{\sOVZ}{0.43ex}{\zeroperp}{<}

\newcommand{\test}[1]{(#1+\mirror{#1})}
\newcommand{\cc}[1]{(\# #1)}
\newcommand{\xx}{C_4}
\newcommand{\tc}[1]{t{#1}}
\newcommand{\ff}[1]{t'{#1}}
\mysymb{\zz}{0.5ex}{-}{\angle}
\newcommand{\multi}[2]{\genfrac{[}{]}{0pt}{}{#1}{#2}}

\newcommand{\sumxy}[4]
{#1 \sum_{\substack{x,y\in #4, x<y}} #2 #3 \\
=\frac 12 #1\left(\left(\sum_{x\in #4} #2\right)^2-\sum_{x\in #4} #2^2\right)}
\newcommand{\sumxynobreak}[4]
{#1 \sum_{\substack{x,y\in #4, x<y}} #2 #3
=\frac 12 #1\left(\left(\sum_{x\in #4} #2\right)^2-\sum_{x\in #4} #2^2\right)}

\newcommand{\EminusUV}{(E\setminus E(u))\setminus (E(v)-(u,v))}

\section{Reduction from Counting 4-Cycles to Quartet Distance}

In this section we provide a sequence of reductions from counting 4-cycles in a multigraph to computing the quartet distance between two trees.
See Figure~\ref{fig:plan} for an overview of the reductions.
Consequently, there is no algorithm for quartet distance that runs significantly faster than in $\Oh(n^{1.48})$ time unless we can count 4-cycles faster.
In particular, existence of an $\Oh(n\log n)$ algorithm for the quartet distance would imply a surprisingly fast algorithm for counting 4-cycles.

As a warm-up, we will show how to reduce counting 4-cycles in a simple graph to counting 4-cycles in a simple bipartite graph.
Then, we will show how to construct for a bipartite simple graph two trees in such a way
that the number of 4-cycles in the original graph can be efficiently extracted
from their quartet distance.

\subsection{Warm-up: From Simple Graphs to Simple Bipartite Graphs}\label{se:simple_to_bipartite}

For a given simple graph $G=(V,E)$ we construct a bipartite graph $G'=(V_1\cup V_2,E')$ such that every $v\in V$ corresponds to two nodes $v_1\in V_1$ and $v_2\in V_2$.
For every edge $\{u,v\}\in E$ we create two edges $\{u_1,v_2\}$ and $\{u_2,v_1\}$ in $G'$.
Then every cycle $(a,b,c,d)$ in $G$ corresponds to two cycles $(a_1,b_2,c_1,d_2)$ and $(a_2,b_1,c_2,d_1)$ in $G'$
and there are no other 4-cycles in $G'$.
We conclude that the complexity of counting 4-cycles in simple graphs is asymptotically the same as in bipartite simple graphs.

\subsection{From Simple Bipartite Graphs to Quartet Distance}\label{se:counting_4c_to_quartet_d}

In this section we show how to reduce counting 4-cycles in a simple bipartite graph to computing the quartet distance between two trees.
We first provide some insight into the structure of 4-edge subgraphs of a bipartite graph which we call shapes.

\paragraph{Properties of shapes.} We first consider all nodes with non-zero degrees in a shape.
For instance, nodes in shape $\sB$ have the following (non-zero) degrees: $3,1$ on the left side and $1,1,2$ on the right side.
We call sorted list of non-zero degrees of $V_1$ (respectively $V_2$) in a shape its left (respectively right) representation.
Then two representations separated by a dash form the representation of a shape.
For instance, the representation of $\sB$ is $3,1-2,1,1$.
There are 5 possible left and right representations: $(4),(3,1),(2,2),(2,1,1)$ and $(1,1,1,1)$.
Next, the representation of a shape almost uniquely determines the shape.
For instance, $3,1-1,1,1,1$ corresponds only to one shape $\sC$.
Note that the only representation which does not uniquely describe a shape is $2,1,1-2,1,1$ as it represents two distinct shapes: $\sG$ and $\sH$.
The notion of representations gives us a systematic way to list all 16 possible shapes.
In~Table~\ref{tab:config_to_reflect} we list 6 of them and omit another 6 shapes which are their mirror reflections, that is they are reflections along the vertical axis.
For example, we say that~$\mirror\sC$ is the mirror reflection of $\sC$.
In~Table~\ref{tab:config_not_reflect} we list all the remaining 4 shapes which remain unchanged under mirror reflection.

\begin{table}[t]
\center
\begin{tabular}{c|c|c|c|c|c}
$\sA$ & $\sB$ & $\sC$ & $\sE$ & $\sF$ & $\sI$ \\
$4-1,1,1,1$ & $3,1-2,1,1$ & $3,1-1,1,1,1$  & $2,2-2,1,1$ & $2,2-1,1,1,1$ & $2,1,1-1,1,1,1$ \\
\end{tabular}
\caption{Six possible shapes which change under mirror reflection.}
\label{tab:config_to_reflect}
\vspace{0.5cm}
\begin{tabular}{c|c|c}
$\sD$ & $\sG \quad , \quad \sH$ & $\sJ$\\
$2,2-2,2$ & $2,1,1-2,1,1$ & $1,1,1,1-1,1,1,1$ \\
\end{tabular}
\caption{Four possible shapes which remain unchanged under mirror reflection.}
\label{tab:config_not_reflect}
\end{table}

\paragraph{The reduction.} On a high level, we design the reduction in such a way that the quartet distance between the constructed trees can be obtained by counting particular shapes in the considered simple graph $G=(V_1\cup V_2,E)$ and adding up the results.
Some of the shapes can be counted in linear time, for instance the number of shapes $\sA$ in $G$ is $\cc\sA=\sum_{v\in V_1} {\deg(v) \choose 4}$. 
However, it is more difficult to compute $\cc\sI$, not to mention $\cc\sD$ which is exactly the sought number of 4-cycles.
We will relate these numbers to $\cc\sD=:\xx$ and then express the quartet distance as a multiple of the number of 4-cycles plus some value that we can compute in linear time.
Solving this simple equation gives us $\xx$.

Given a bipartite graph $G=(V_1\cup V_2,E)$ we construct the trees $T_1$ and $T_2$ in the following way.
Tree $T_i$ consists of nodes representing all non-isolated nodes from $V_i$ attached to the root and nodes representing edges from $E$ attached to the node corresponding to their endpoint from $V_i$.
Note that there is exactly one such node, as $G$ is bipartite and there is a bijection between the leaves of $T_i$ and $E$.
See Figure~\ref{fig:graph_to_trees} for an example.

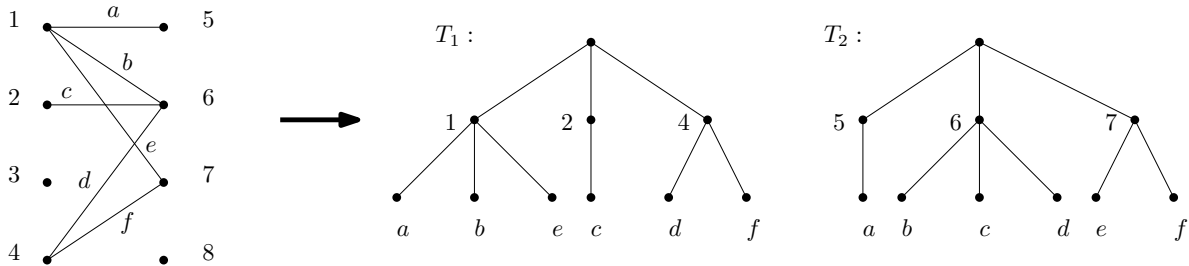
\FIGURE{b}{.95}{graph_to_trees}{Instance of the quartet distance problem obtained from the bipartite graph on the left.}

\paragraph{Quartets.} Recall that in the quartet distance between trees we consider subtrees induced by four leaves.
The above construction guarantees that the subtree of $T_1$ (respectively $T_2$) induced by a set of four leaves $L=\{e_1,e_2,e_3,e_4\}$ is uniquely determined by the left (respectively right) representation of the graph consisting of edges $\{e_1,\ldots,e_4\}$.
See Figure~\ref{fig:shapes_in_one_tree}.

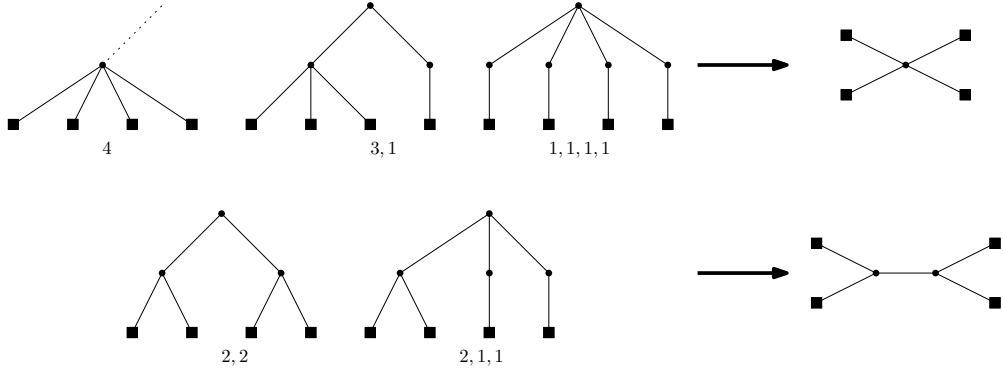
\FIGURE{t}{.8}{shapes_in_one_tree}{Left: all five possible representations and their corresponding trees.
Right: the corresponding tree topology is either a star (upper row) or a butterfly (lower row).}

As the quartet distance between $T_1$ and $T_2$, denoted as $\QD$, is the number of sets of four leaves that are not related by the same topology in both trees, $\QD$ equals $\#\text{ of leaves} \choose 4$ minus the number of subsets of four leaves that are related by the same topology in both trees.
From now on we will focus on computing only the latter number.
The agreeing topologies can be either stars (unresolved quartets, upper row in Figure~\ref{fig:shapes_in_one_tree}) or butterflies (resolved quartets, bottom row in Figure~\ref{fig:shapes_in_one_tree}).
In stars the order of labels on the leaves do not matter, so it is enough that the quartet induces a star in both trees.
There are five shapes which induce a star in both trees: $\sA, \mirror{\sA}, \sC, \mirror{\sC}$ and $\sJ$.
Next, in order to ensure that a quartet induces a butterfly in both trees, its left and right representations must be either $2,2$ or $2,1,1$.
As the labels on leaves do matter for butterflies, among all the 5 shapes inducing them (recall that there are two shapes represented by $2,1,1-2,1,1$), only $\sG$ has matching labels on leaves.
See Figure~\ref{fig:five_butterflies}.
Summarising, we obtain the following equality:
$$\binom{\#\text{ of leaves}}{4}-\QD= \left( \cc\sA + \cc{\mirror{\sA}} +\cc\sC +\cc{\mirror{\sC}}+\cc\sJ \right)+ \cc\sG$$

\newcommand{\sEmirror}{\mirror\sE}
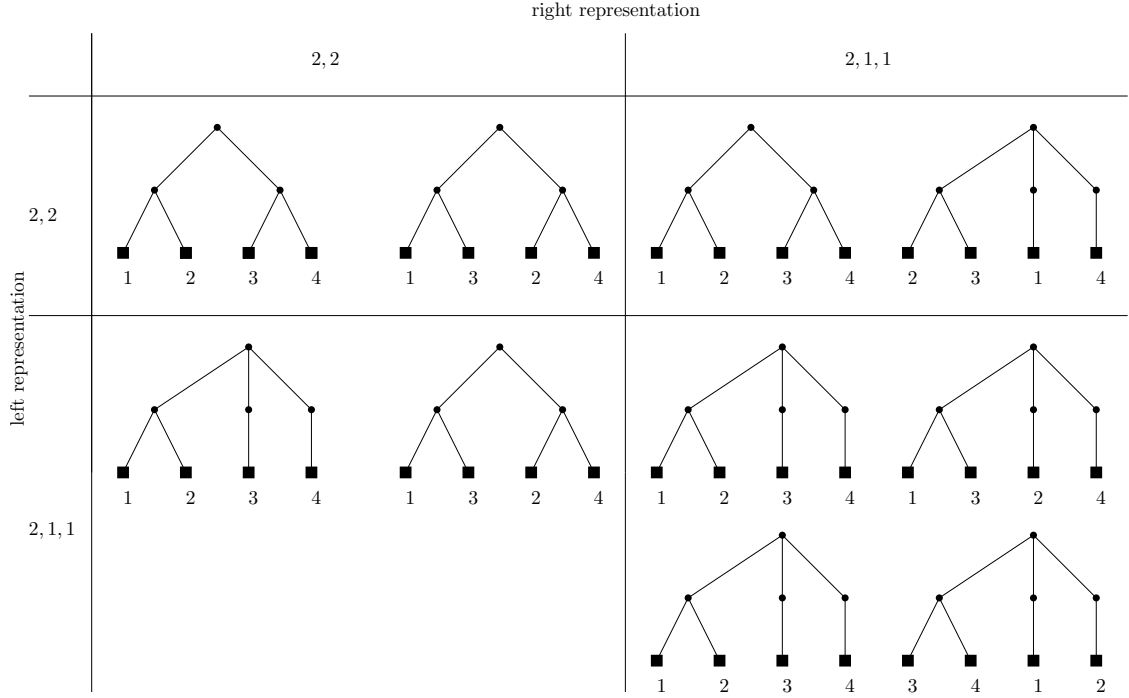
\FIGURE{h}{.9}{five_butterflies}{There are 5 different ways of how a quartet of leaves can induce a butterfly simultaneously in both trees. They correspond to the following 5 shapes (starting from the upper left corner in the clockwise order): $\protect\sD,\protect\sE,\protect\sH,\protect\sG,\protect\sEmirror$.
Among them, only $\protect\sG$ has matching labels.}

Hence, we need to compute $\cc\sA,\cc{\mirror{\sA}},\cc\sC,\cc{\mirror{\sC}},\cc\sJ$ and $\cc\sG$ to obtain $\QD$.
In the following lemmas we show that all the above values except for $\cc\sJ$ can be computed in linear time, whereas $\cc\sJ$ is directly related to the number of 4-cycles in $G$.
More precisely, for every shape $R$ we will express its corresponding value as $\cc{R}=t_{R}+d_R\xx$, where $t_R$ is an auxiliary value which can be computed from the considered bipartite graph $G$ in linear time and $d_R$ is a constant.
For instance, $\cc\sA=\tc\sA$ means that $\cc\sA$ can be obtained by computing a certain auxiliary value in linear time.
The main lemma of this section is that $\cc\sJ=\tc{\sJ} + \xx$ which implies that we can compute $\cc\sJ$ from the number of 4-cycles and vice versa in linear time.

Let $m=|E|$, $d(u)$ denotes degree of the node $u$ and $N(u)$ the set of its neighbors.
As for now we assume that $E \subseteq V_1 \times V_2$ (recall that the graph is bipartite) and $E$ consists of ordered pairs $(u,v)$. $u$ denotes a node from $V_1$ and $v$ from~$V_2$.
\begin{lemma}\label{le:calculations_easy}
 $\cc\sA,\cc\sB,\cc\sC$ and $\cc\sG$ can be computed from $G$ in linear time.
\end{lemma}
\begin{proof} 
We compute the first three values directly:
\begin{enumerate}
 \item $\cc\sA=\sum_{u\in V_1}\binom{d(u)}{4}$
 \item $\cc\sB=\sum_{(u,v)\in E}\binom{d(u)-1}{2}(d(v)-1)$
 \item $\cc\sC=\left(\sum_{u\in V_1}\binom{d(u)}{3}(m-d(u))\right)-\cc\sB$
\end{enumerate}
 \noindent
Calculations for the mirror reflections of the above shapes are symmetric. 
In order to compute $\cc\sG$ we need auxiliary values: $\cc>=\sum_{v\in V_2}\binom{d(v)}{2}$ and $\cc\zz=\sum_{(u,v)\in E}(d(u)-1)(d(v)-1)$. Then:
$$\cc\sG=\frac 12 \left(\sum_{(u,v)\in E}(d(u)-1)\left(\cc>-\binom{d(v)}{2}\right)-\cc{\mirror\sB}-\cc{\zz}-2\cc\sB\right)$$

We derive the above formula in steps.
For each edge $(u,v)$ we count shapes in which the edge is one of the sides of $<$ in $\sG$.
See Figure~\ref{fig:example}(a) with names of all the nodes in $\sG$.
First, we have $(d(u)-1)$ possibilities for the node $w$ which is incident to $u$, but different than $v$.
Second, we need to account for the $>$ parts of $\sG$ that do not have the corner in node $v$.
There are $\left(\cc>-\binom{d(v)}{2}\right)$ of them and we obtain $\sum_{(u,v)\in E}(d(u)-1)\left(\cc>-\binom{d(v)}{2}\right)$.

\newcommand{\sBmirrored}{\mirror\sB}
\newcommand{\sZmirrored}{\mirror\zz}

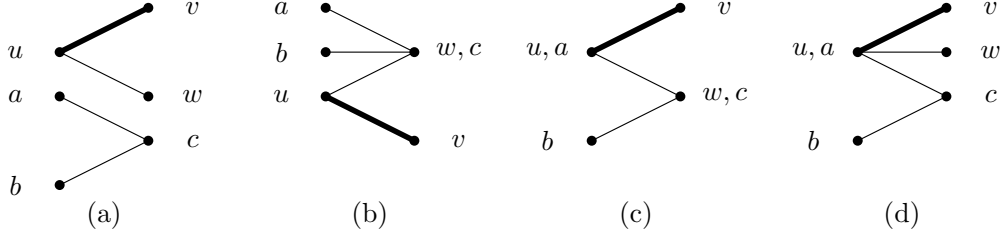
\FIGURE{h}{0.8}{example}{While computing $\cc{\protect\sG}$ we iterate over all edges $(u,v)$. (a) Naming of
vertices, (b),(c),(d)~subtracted shapes corresponding to respectively $\protect\sBmirrored,\protect\sZmirrored,\protect\sB$.}

Now we have counted too many shapes, because we did not ensure that the node $c$ is different than $w$ and that both $a$ and $b$ are different than $u$.
To account for $w=c$ we subtract~$\cc{\mirror\sB}$ (when both $a$ and $b$ are different than $u$, see Figure~\ref{fig:example}(b)) and $\cc{\mirror\zz}$ (when $a$ or $b$ coincide with $u$, see Figure~\ref{fig:example}(c)).
Next, for the case when $c\ne w$ we subtract $2\cc\sB$ for the case when $a$ or $b$ coincide with $u$, see Figure~\ref{fig:example}(d).
This term is multiplied by 2, because we counted it both for the distinguished edge $(u,v)$ and $(u,w)$.

Finally, we need to divide the whole expression by 2, because every shape $\sG$ is counted twice, both for the distinguished edge $(u,v)$ and $(u,w)$.
\end{proof}
\noindent
To sum up, the above lemma implies that the number of shapes $\sA,\sB,\sC,\sG$ can be computed in linear time from $G$.
We do not provide calculations for the mirror reflections of these shapes ($\mirror\sA,\mirror\sB$ and $\mirror\sC$), as it suffices to rewrite all the expressions replacing nodes from $V_1$ with $V_2$ and vice-versa. Similarly, we do not mention the mirror reflections in the following lemmas.

\newcommand{\sumxye}[5]
{#1 \sum_{\substack{x,y\in #4, x<y}} #2 #3 #5
=\frac 12 #1\left(\left(\sum_{x\in #4} #2\right)^2-\sum_{x\in #4} #2^2\right)}

\begin{lemma}\label{le:calculations_hard}
 The following equalities hold:
 \begin{enumerate}
  \item $\cc\sE=\tc{\sE}-2\xx$
  \item $\cc\sF=\tc{\sF}+\xx$
  \item $\cc\sH=\tc{\sH}+4\xx$
  \item $\cc\sI=\tc{\sI}-2\xx$
 \end{enumerate}
\end{lemma}
\begin{proof} Let $\ff{_R}$ be an auxiliary variable used to express $t_R$.
 \begin{enumerate}
  \item  Let $\tc\sE=\sumxy{\sum_{v\in V_2}}{(d(x)-1)}{(d(y)-1)}{N(v)}$.\\
  Then: $\cc\sE=\tc\sE -2\cc\sD=\tc\sE-2\xx$.
  
  \item Let $\ff\sF=\sumxye{}{\binom{d(x)}{2}}{\binom{d(y)}{2}}{V_1}{}$.\\
  Then:
  $\cc\sF=\ff\sF - \cc\sE -\cc\sD=\ff\sF - (\tc{\sE}-2\xx) - \xx = \tc{\sF}+\xx$.
  
  \item 
  Let $\ff\sH=\sum_{(u,v)\in E}(d(u)-1)(d(v)-1)(m-d(u)-d(v)+1)$. \\
  Then: $\cc\sH=\ff\sH-2\cc\sE -2\cc{\mirror{\sE}} - 4\cc\sD \\
  =\ff\sH -2(\tc{\sE}-2\xx)-2(\tc{\mirror\sE}-2\xx)-4\xx =\tc{\sH}+4\xx$. 
 
  \item Let $\ff\sI=\sum_{(u,v)\in E}(d(u)-1)\binom{m-d(u)-d(v)+1}{2}$.\\  
  Then:  $\cc\sI=\frac 12 \left(\ff\sI-2 \cc\sG - \cc{\mirror\sB}-\cc\sH-2\cc\sE-4\cc\sF \right)\\
  =\frac 12 \left( \ff\sI- 2\tc{\sG} - \tc{\mirror\sB} -(\tc{\sH}+4\xx) - 2(\tc{\sE}-2\xx) -4(\tc{\sF}+\xx) \right)\\
  =\tc{\sI}+\frac 12 \left(-4\xx+ 4\xx-4\xx \right)
  =\tc{\sI}-2\xx$. \qedhere
 \end{enumerate}
  
\end{proof}

\begin{lemma}\label{le:m4}
$\cc\sJ=\tc{\sJ} + \xx$, where $\tc{\sJ}$ can be computed from $G$ in $\Oh(|E|)$ time.
\end{lemma}
\begin{proof} 
We first compute $\cc\zz=\sum_{(u,v)\in E} (d(u)-1)(d(v)-1)$\\  and $\cc\leq={\frac 12\left(\sum_{(u,v)\in E} (d(u)-1)(m-d(u)-d(v)+1)-\cc\zz\right)}$. Then:
 \begin{enumerate}
  \item $\cc\equiv = \frac 13 \big(\sum_{(u,v)\in E}\binom{m-d(u)-d(v)+1}{2} - \cc\leq - \cc\geq  \big)$
  \item $\cc\sJ = \frac 14 \left((m-3)\cc\equiv -\cc\sH-2\cc\sI-2\cc{\mirror\sI}  \right)\\
 =\frac 14 \left((m-3)\tc\equiv -(\tc{\sH}+4\xx)-2(\tc{\sI}-2\xx)-2(\tc{\mirror\sI}-2\xx)  \right)\\
 =\tc{\sJ} +\frac 14 ( -4\xx+4\xx+4\xx)=\tc{\sJ}+\xx$ \qedhere
 \end{enumerate}
\end{proof}

\begin{theorem}
 Counting 4-cycles in a graph with $m$ edges can be reduced in linear time to computing the quartet distance between two trees on $m$ leaves.
\end{theorem}

\subsection{From Multigraphs to Simple Graphs}

In this section we show how to count 4-cycles in multigraphs using a polylogarithmic number of black-box calls to counting 4-cycles in simple graphs.
These reductions do not require the graphs to be bipartite.
For the sake of simplicity, we did not optimize constants in these proofs.

\subsubsection{Small Edge Multiplicities.}\label{se:small_mult}
First we consider the case in which every edge in the input multigraph $G=(V,E,\mult)$ has multiplicity bounded by some constant $c$.
To present our reduction we slightly extend the $\mult$ function and allow that it also returns $\star$ as the multiplicity of an edge.
We will call such $\star$-edges \textit{special}.
In terms of counting 4-cycles, these edges have multiplicity 1.
So in this subsection we have $\mult: E\rightarrow\{1,2,\ldots,c,\star\}$.
We start with showing how to count 4-cycles with a particular number of special edges in such multigraphs using black-box calls to the procedure counting 4-cycles in multigraphs with no special edges.
Formally, we need to compute $y^\star(G)=(y_0,y_1,y_2,y_3,y_4)$ where $y_k$ is the number of 4-cycles in $G$ with exactly $k$ special edges.

\begin{lemma}\label{le:cycles_special_edges}
 For a graph $G=(V,E,\mult)$ with special edges and edge multiplicities bounded by $c$, we can compute $y^\star(G)$ in linear time with a constant number of black-box calls to counting 4-cycles in multigraphs with edge multiplicities bounded by $c$ and of asymptotically the same size as $G$.
\end{lemma}
\begin{proof}
 For a constant $q$, let $\Phi_q(G)$ be the graph obtained from $G$ in the following way:
 \begin{itemize}
  \item For every node $v\in V$ we create $q$ nodes $v^{(i)}$, for all $0\leq i<q$.
  \item For every non-special edge $e=\{u,v\}\in E$ ($\mult(e)\ne \star$), we create $q$ edges $e^{(i)}=\{u^{(i)},v^{(i)}\}$ between the $i$-th copies of nodes $u$ and $v$, for all $0\leq i<q$.
  Each such edge preserves the multiplicity of the original edge, $\mult(e)$.
  \item For every special edge $e=\{u,v\}\in E$ ($\mult(e)= \star$), we create $q^2$ edges $e^{(i,j)}=\{u^{(i)},v^{(j)}\}$ between the $i$-th copy of $u$ and $j$-th copy of $v$, for all $0\leq i,j<q$.
  Each such edge has multiplicity $1$.
 \end{itemize}
 
\FIGURE{h}{0.7}{bad_cycles_corrected}{Two kinds of bad cycles introduced in $\Phi_q(G)$, they can pass through nodes corresponding to either two or three distinct nodes from $G$.}
 
 Observe that every 4-cycle in $G$ with either 0 or 1 special edge corresponds to $q$ cycles in $\Phi_q(G)$.
 Similarly, 4-cycle with $k$ special edges in $G$ corresponds to $q^k$ cycles in $\Phi_q(G)$ for $k>1$.
 These are the only 4-cycles in $\Phi_q(G)$ without two copies of the same node.
 However, in $\Phi_q(G)$ we also introduced 4-cycles passing through two copies of one node (we call them ``bad cycles''), see Figure~\ref{fig:bad_cycles_corrected}.
 Notice that each special edge introduces $\binom{q}{2}^2$ cycles passing through nodes corresponding to exactly two distinct nodes of $G$.
 Let $\spec(v)$ be the number of special edges incident to the node $v$ in $G$.
Observe that there are $\sum_{v\in V} \binom{\spec(v)}{2}\binom{q}{2}q^2$ cycles in $\Phi_q(G)$ passing through nodes corresponding to exactly three distinct nodes of $G$.
 Hence we can count all the bad cycles in linear time.
  Recall that $y_k$ is the number of 4-cycles in $G$ with exactly $k$ special edges, so we have:
 $$\#C_4\left(\Phi_q(G)\right) = \sum_{k=0}^4 y_k\cdot q^{\max(k,1)} + \#\text{bad cycles}(\Phi_q(G))$$
 
 Notice that for any $q$, $\Phi_q(G)$ has no special edges, so we can compute $\#C_4(\Phi_q(G))$ using one black-box call to counting 4-cycles in a multigraph with no special edges, with multiplicities bounded by $c$, $q|V|$ nodes and at most $q^2|E|$ edges.
 Next, $y_0$ can be retrieved by counting 4-cycles with $G$ restricted to the non-special edges and similarly $y_4$ by considering the graph consisting of only the special edges from $G$.
 Hence we are left with 3 unknown variables $y_1,y_2,y_3$. 
 By computing $\#C_4(\Phi_q(G))$ for three distinct values of $q$, e.g. 1,2 and 3, we obtain the following system of linear equations:
 
 \newcommand{\EqFromMultiToSimple}[2]{y_1\cdot 1 + y_2\cdot #1 +y_3\cdot #1^2 &=#2\#C_4(\Phi_{#1}(G))- \#\text{bad cycles}(\Phi_{#1}(G))-y_0\cdot#1 - y_4\cdot#1^4}
 
 \begin{align*}
 \EqFromMultiToSimple{1}{\ \ \ }\\
 \EqFromMultiToSimple{2}{\frac12(})\\
 \EqFromMultiToSimple{3}{\frac13(})
 \end{align*}

 \noindent
These equations are linearly independent, so we can solve the system for $y_{1},y_{2},y_{3}$.
\end{proof}

Using the above lemma we describe a recursive approach for counting 4-cycles in multigraphs with edge multiplicities bounded by $c$, that at the bottom level uses a black-box call to counting 4-cycles in simple graphs.
We first make all edges with multiplicity $c$ special and apply Lemma \ref{le:cycles_special_edges} for the graph with special edges and multiplicities bounded by $c-1$.
We stress that inside Lemma~\ref{le:cycles_special_edges} we create graphs with no special edges and with edge multiplicity bounded by $c-1$.
Finally, by grouping the cycles by the number of special edges, we count the number of 4-cycles in the original graph.

\begin{algorithm}[h]
\begin{algorithmic}[1]
  \Function{Count$C_4$}{$G=(V,E,\mult),c$}
  \If{$c=1$}
    \State \Return $\#C_4(G)$ \Comment{this graph is simple}
  \EndIf
  \State $\mult'(e):=\begin{cases}
    \star&, \mult(e)=c\\
    \mult(e)&,\mult(e)<c
    \end{cases}$
  \State $G':=(V,E,\mult')$ \Comment{with special edges and multiplicities bounded by $c-1$}
  \State $y:=y^\star(G')$ \Comment{applying Lemma \ref{le:cycles_special_edges} that uses \textsc{Count}$C_4(\cdot, c-1)$}
  \State \Return $\sum_k c^k\cdot y_k$
  \EndFunction
\end{algorithmic}
\caption{The recursive procedure of counting 4-cycles in multigraphs.}
\label{alg:naive_coloring}
\end{algorithm}

Observe that the total number of recursive calls is exponential in $c$ and in every recursive call the graph increases by a constant factor.
As $c$ is a constant, all the time we operate on graphs asymptotically of the same size as the original multigraph and we consider a constant number of them.
We remark that the exponential dependency on $c$ can be avoided by applying a slightly more complex procedure.
However, this is irrelevant for the statement of the following corollary.

\begin{corollary}\label{cor:small_to_simple}
 We can count 4-cycles in a multigraph $G$ with edge multiplicities bounded by a constant $c$ using a constant number of black-box call to counting 4-cycles in a simple graph of asymptotically the same size as $G$.
\end{corollary}

\subsubsection{Multiplicity Reduction.}\label{se:mult_reduction}

Recall that we consider arbitrary multigraphs, not necessarily bipartite.
First consider a simple graph and a coloring of its edges $K: E\rightarrow \{1,2,3,4,\perp\}$ in which every edge is colored with one of the four colors or not colored at all.
We define $f_K(a,b,c,d)$ as the number of 4-cycles consisting of exactly $a$ edges of the first color, $b$ edges of the second color etc. From now on the variables $a,b,c,d$ always satisfy $0\le a,b,c,d \le 4$ and $a+b+c+d=4$.

\begin{lemma}\label{le:exact_count_of_4_colored_graph}
 For every coloring $K$ of edges of a simple graph $G$, we can compute $f_K(a,b,c,d)$ for all quadruples $a,b,c,d$ in a constant number of black-box calls to counting 4-cycles in multigraphs of asymptotically the same size as $G$, but with multiplicities of edges bounded by a constant.
\end{lemma}
\begin{proof}
 Let $g_K(d_1,d_2,d_3,d_4)$ be the number of 4-cycles in the multigraph $G_K(d_1,d_2,d_3,d_4)$ constructed from the graph $G$ and the coloring $K$ in such a way that all edges with color $i$ have multiplicity $d_i$ and all other edges have multiplicity 0.
 Then all 4-cycles can be grouped by the number of edges of each color and the following holds:
 
 $$g_K(x,y,z,t)=\sum_{\substack{0\le a,b,c,d\le 4\\a+b+c+d=4}} f_K(a,b,c,d)x^ay^bz^ct^d$$
 
 For a specific coloring $K$, $g_K(x,y,z,t)$ is a polynomial in variables $x,y,z$ and $t$ with coefficient $f_K(a,b,c,d)$ by $x^ay^bz^ct^d$.
 We need to determine 35 unknown variables $f_K(a,b,c,d)$ for all quadruples $(a,b,c,d)$ such that $0\le a,b,c,d\leq 4$ and $a+b+c+d=4$.
 Observe that we can evaluate $g_K(x,y,z,t)$ for any values $x,y,z,t$ directly, without knowing its coefficients, using the procedure for multigraphs with bounded multiplicities.
 We would like to evaluate $g_K(x,y,z,t)$ in many points and then interpolate it.
 We will carefully choose the values of $x,y,z$ and $t$ so that we can interpolate $g_K(x,y,z,t)$ using the standard univariate polynomial interpolation.
 Consider the following polynomial:
 $$h_K(x):=g_K(x,x^5,x^{25},x^{125})=\sum_{\substack{0\le a,b,c,d\le 4\\a+b+c+d=4}} f_K(a,b,c,d)x^{a+5b+25c+125d}$$
 $h_K(x)$ is a univariate polynomial of constant degree 500 in which all the non-zero coefficients correspond to the value of $f_K(a,b,c,d)$ for some $a,b,c,d$. 
 Hence we can evaluate $h_K(x)$ for $x=1,2,\ldots,501$ and then interpolate.
 \end{proof}

Recall that now all the edge multiplicities are bounded by a parameter $U$.
A naive application of the above lemma would be to consider all quadruples of different multiplicities $(a_1,\ldots,a_4)$ such that $a_i \in [1,2,\ldots,U]$ and $a_i>a_{i+1}$ for all $i\in\{1,2,3\}$.
For each of them we color the edges of~$G$ in such a way that edges with multiplicity $a_i$ get color $i$ and all other edges get $\perp$ (no color).
For convenience, we denote such a coloring simply by a vector with individual multiplicities $[a_1,\ldots,a_4]$.
If the vector has less than four elements, then some colors are not used at all and by default, the last arguments of the function $f$ are set to 0.
For each of the considered colorings we add to the total number of 4-cycles the value of $a_1a_2a_3a_4\cdot f_{[a_1,a_2,a_3,a_4]}(1,1,1,1)$ as we can choose any of the $a_i$ copies of the edge with color $i$.
We also need to add terms corresponding to cycles in which some edges have the same color.
However, we need to be careful in order not to count more than once any 4-cycle with less than 4 colors. 
Algorithm~\ref{alg:naive_coloring} shows how to count every 4-cycle exactly once.
Clearly this approach is too slow for our purposes, because it uses $\Omega(U^4)$ black-box calls.

\begin{algorithm}[h]
\begin{algorithmic}[1]
  \State $C_4:=0$
  \For{$i=1,2,\ldots,U$}
    \State $C_4\pluseq i^4\cdot f_{[i]}(4)$
    \For{$j=1,2,\ldots,i-1$}
      \State $C_4\pluseq i^2j^2\cdot f_{[i,j]}(2,2)+i^3j\cdot f_{[i,j]}(3,1)+ij^3\cdot f_{[i,j]}(1,3)$
      \For{$k=1,2,\ldots,j-1$}
	\State $C_4\pluseq i^2jk\cdot f_{[i,j,k]}(2,1,1)+ij^2k\cdot f_{[i,j,k]}(1,2,1)+ijk^2\cdot f_{[i,j,k]}(1,1,2)$
	\For{$l=1,2,\ldots,k-1$}
	  \State $C_4\pluseq ijkl\cdot f_{[i,j,k,l]}(1,1,1,1)$
	\EndFor
      \EndFor
    \EndFor
  \EndFor
\end{algorithmic}
\caption{A naive way of coloring edges for counting 4-cycles.}
\label{alg:naive_coloring}
\end{algorithm}

In order to further reduce the complexity for $U\geq 2$, on a high level, we will iterate only over all powers of two that appear in the binary representation of each multiplicity.

\begin{lemma}
 We can count 4-cycles in a multigraph with edge multiplicities bounded by $U$ with $\Oh(\log^4 U)$ black-box calls to counting 4-cycles in colored simple graphs.
\end{lemma}
\newcommand{\QQ}{Q}
\begin{proof}
 Let $G$ be the input multigraph and $G'$ a graph in which we split every edge $e$ from $G$ into up to $\log U$ edges with multiplicities that are distinct powers of two and sum up to $\mult(e)$.
 Every edge multiplicity in $G'$ is a power of two, so to distinguish this from arbitrary multiplicities in $G$, we say that edges in $G'$ have weights.
 See Figure~\ref{fig:big_edge_mult}.
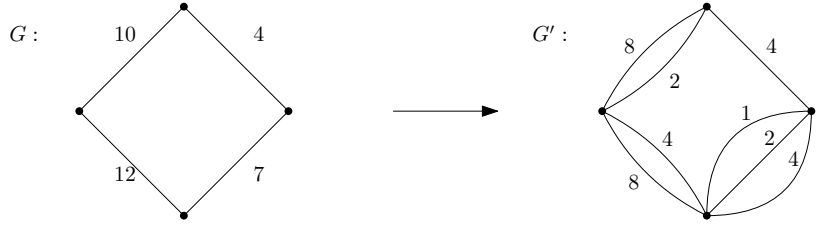
 \FIGURE{h}{.65}{big_edge_mult}{We split every edge $e$ into up to $\log U$ edges with weights that are powers of 2.}
 
 Consider a single 4-cycle $C$ in $G'$ consisting of edges with the multiset of edge weights $W=\{p_1,p_2,p_3,p_4\}$ where every $p_i$ is a power of two and $1\le p_1\le p_2 \le p_3 \le p_4 \le U$.
 Notice that $C$ corresponds to $\prod_i p_i$ 4-cycles in $G$.
 Now we would like to count all 4-cycles in $G'$ with the multiset of weights~$W$.
 Let $Q=\{q_1,\ldots,q_d\}$ be the set of all distinct weights in $W$.
 
 Suppose first that no pair of nodes in $G$ is connected with more than one edge with weight from~$Q$.
 Then we can color every edge of $G'$ with at most one of four colors and use Lemma~\ref{le:exact_count_of_4_colored_graph} to count in $G'$ all 4-cycles with edge weights $W$.
 However, we cannot use this approach when there is more than one edge between a pair of nodes in $G'$, because then some edges would be multi-colored, but Lemma~\ref{le:exact_count_of_4_colored_graph} allows only simple graphs.
 
 To overcome this difficulty, we need to divide all 4-cycles in $G'$ with the multiset of weights~$W$ into smaller groups and count each of them separately.
 For every edge in the cycle we will specify not only a weight $p_i$ but the whole set $M_i\ni p_i$
 equal to the intersection of $Q$ and the set of all weights of edges connecting the corresponding pair of nodes.
 In order to ensure that no two groups of cycles overlap, we require that if $p_i=p_{i+1}$ then $M_i \preceq M_{i+1}$ where $\preceq$ denotes lexicographic order.
 To sum up, in order to count all 4-cycles in $G'$ with the multiset of weights $W=\{p_1,p_2,p_3,p_4\}$ (recall that $p_i\leq p_{i+1}$) we divide them into groups by the choice of four sets $M_i$ that satisfy:
 \begin{itemize}
  \item for $i\in\{1,2,3,4\}: p_i\in M_i \subseteq Q$,
  \item for $i\in\{1,2,3\}:p_i=p_{i+1} \implies M_i \preceq M_{i+1}$.
 \end{itemize}
 Now we assign an edge $e$ the color $M_i$ if and only if $M_i=\text{BinaryRepresentation}(\mult(e)) \cap Q$.
 As we obtain a simple colored graph, we can count 4-cycles such that the multiset of colors of their edges is $\{M_1,M_2,M_3,M_4\}$ using Lemma~\ref{le:exact_count_of_4_colored_graph}.
 Let $\{N_1,N_2,\ldots,N_f\}$ be the set of all distinct colors in $\{M_1,M_2,M_3,M_4\}$.
 Next, because some of the colors $M_i$ might be equal (or equivalently, $f<4$) we need to multiply the obtained number of 4-cycles by:
 $$\prod_{M\in\{N_1,N_2,\ldots,N_f\}}\binom{x_{M}}{y_{M,1},y_{M,2},\ldots ,y_{M,d}}$$
 where the product might be over less than four elements if some of $M_i$ are equal.
 The multiplied expressions are multinomial coefficients, $x_M=|\{i:M_i=M\}|$ is the number of colors $M_i$ equal to $M$ and $y_{M,j}=|\{i: M_i=M \wedge p_i=q_j\}|$ is the number of weights $p_i=q_j$ that are assigned the color $M_i=M$. Recall that $Q=\{q_1,\ldots,q_d\}$ is the set of all distinct weights from $W$.
 
 To conclude, in order to count all 4-cycles with the multiset of weights $W$ we iterate over all choices of four subsets $M_i$ satisfying the above properties, construct an appropriate simple colored graph, count particular 4-cycles in it and multiply the obtained number by the multinomial coefficients.
 We sum these numbers up and multiply by $\prod_i p_i$, the number of cycles in $G$ that correspond to a cycle in $G'$ with the multiset of weights $W$.
 
 We need to repeat this approach for all choices of the multiset~$W$, so in total there will be at most $\log^4U\cdot16^4=\Oh(\log^4U)$ black-box calls to counting 4-cycles in simple colored graphs.
\end{proof}

\noindent
Combining the above lemma with Lemma~\ref{le:exact_count_of_4_colored_graph} and Corollary~\ref{cor:small_to_simple} we obtain:

\begin{corollary}\label{cor:c4_in_multigraphs}
 We can count 4-cycles in a multigraph $G$ with edge multiplicities bounded by $U$ in $\Oh(\log^4 U)$ black-box calls to counting 4-cycles in simple graphs of asymptotically the same size as $G$.
\end{corollary}

\section{From Quartet Distance to Counting 4-Cycles}\label{se:from_qd_to_4c}

In the previous section we proved, that computing the quartet distance is at least as hard as counting 4-cycles.
Now we will show how to use state-of-the-art algorithm for counting 4-cycles to compute quartet distance faster.
As in Section~\ref{se:counting_4c_to_quartet_d}, we will count quartets of leaves related by the same topology in both trees.
Recall that there are two possible topologies: resolved quartet (butterfly) and unresolved quartet (star).
We will count shared resolved quartets using $\Oh(n\log n)$ algorithm of Brodal~et~al. \cite{BrodalFMPS13} (value $A$ computed in Section 7.3 there).

For counting shared unresolved quartets (stars), we develop a new algorithm which reduces the original question to counting 4-cycles in many different multigraphs.
To provide an intuition, we first describe a slow approach.
Every star has a central node, so we iterate over all central nodes $c_1\in T_1$ and $c_2\in T_2$.
Then we need to count quartets of leaves such that they are in different subtrees connected to $c_1$ and $c_2$. 
Observe that this corresponds to the number of matchings of size~4~($\sJ$) in the multigraph in which left (respectively right) nodes correspond to subtrees connected to the node $c_1$ (respectively $c_2$), and the multiplicity of an edge $(a,b)$ is the number of common leaves in the $a$-th subtree connected to $c_1$ and the $b$-th subtree connected to $c_2$.
See Figure~\ref{fig:star_to_m4}.

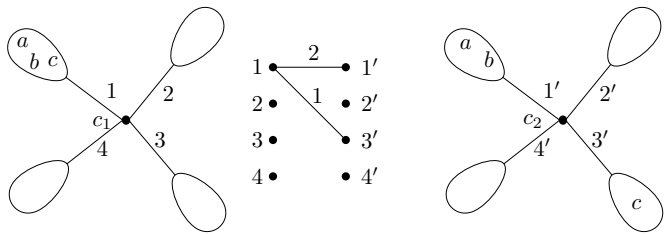
\FIGURE{b}{.53}{star_to_m4}{Construction of a bipartite multigraph for two central nodes $c_1\in T_1$ and $c_2\in T_2$.}

Similarly as in Lemma~\ref{le:m4}, we can count matchings of size~$4$ by counting 4-cycles (the proof follows roughly the same idea
but requires more calculations, and can be found in Appendix~\ref{se:missing_proofs}):

\begin{theorem}\label{thm:counting_multi}
In any bipartite multigraph $G$ we have $\cc\sJ=\tc{\sJ} + \xx$, where $\tc{\sJ}$ can be computed from $G$ in $\Oh(|E|)$ time.
\end{theorem}

Hence we reduced computing quartet distance to $\Oh(n^2)$ black-box calls to counting 4-cycles in a possibly large multigraph, which, from Corollary~\ref{cor:c4_in_multigraphs}, can be done with $\Oh(\log^4U)$ black-box calls to the procedure counting 4-cycles in a simple graph.
To obtain a faster algorithm for computing quartet distance we need to decrease both the number of black-box calls and have some control on the total size of the constructed multigraphs.
In Section~\ref{se:faster_alg_for_qd} we design a divide and conquer approach based on top tree decomposition that, combined with the state-of-the-art algorithm for counting 4-cycles, allows us to improve the state-of-the-art algorithm for quartet distance:

\begin{theorem}\label{thm:quartet_final}
An algorithm for counting 4-cycles in a simple graph with $m$ edges in $\Oh(m^{\pow})$ time
implies an algorithm for computing the quartet distance between trees on $n$ leaves in $\Ohtilda(n^\pow)$ time.
\end{theorem}
\noindent
Now we can plug in the algorithm of Vassilevska Williams et al. \cite{WilliamsWWY15} for counting 4-cycles in $\Oh(m^{2-\frac{3}{2\omega+1}})=\Oh(m^\gamma)$ time.
Because $\omega<2.373$ \cite{Gall14a,Williams12}, we have $\pow<1.48$ and obtain an algorithm for computing the quartet distance between two trees on $n$ leaves in $\Oh(n^{1.48})$ time.
Furthermore, we can also use the $\Oh(n^\omega)$ algorithm by Alon et al. \cite{AlonYZ97} for counting 4-cycles whenever the graph
is dense. As described in more detail at the end of the next section, by appropriately switching between both approaches we obtain
an algorithm computing the quartet distance in $\tilde\Oh(\min(n^{1.48},nd^{0.77}))$ time, thus showing Theorem~\ref{thm:new_alg_for_qd}.

\section{Faster Algorithm for Quartet Distance}\label{se:faster_alg_for_qd}

In this section we describe a faster algorithm for computing the quartet distance.
The starting point is that, as observed in Section~\ref{se:from_qd_to_4c}, we only need to count quartets of leaves that induce stars in both trees, which we call shared stars.
We group the stars by their central nodes and then count quartets using the procedure for counting 4-cycles applied to many small bipartite multigraphs.
However, this approach is too slow, because there are $\Theta(n^2)$ pairs of central nodes to consider.
To bypass this difficulty, we will consider some of the pairs of central nodes explicitly (there will be $\Oh(n\log^2 n)$ of such explicitly considered
pairs) and then process the remaining ones aggregately in bigger groups.
\paragraph{Top trees.} We root both trees at arbitrarily chosen leaves and then apply a hierarchical decomposition based on top trees introduced by Alstrup et al. \cite{AlstrupHLT05} and then extended by Bille et al.~\cite{BilleGLW15}.
A top tree of a tree $T$ is an ordered and labeled binary tree describing a hierarchical decomposition
of~$T$ into clusters where each cluster is either a single edge, or obtained by merging two clusters.
Each cluster has at most two boundary nodes and there are five possible ways of merging two clusters, see Figure~\ref{fig:toptrees} (Figure 2 in \cite{BilleGLW15}).
Let the merged boundary node of a cluster $C$ be the common node of two clusters that form $C$.
There are $\Oh(n)$ clusters describing each of the trees, as we start with $n-1$ clusters for each edge and every merge decreases the number of clusters by one.
As we root the trees at leaves, the final cluster will have at most two boundary nodes which are leaves of the original tree.
Let $\TT_i$ be the top tree representing tree $T_i$.
\begin{property}[Corollary 1 from \cite{BilleGLW15}]\label{prop:log_height_of_toptree}
Given a top tree on $n$ nodes, we can create its top tree of height $\Oh(\log n)$ in $\Oh(n)$ time.
\end{property}

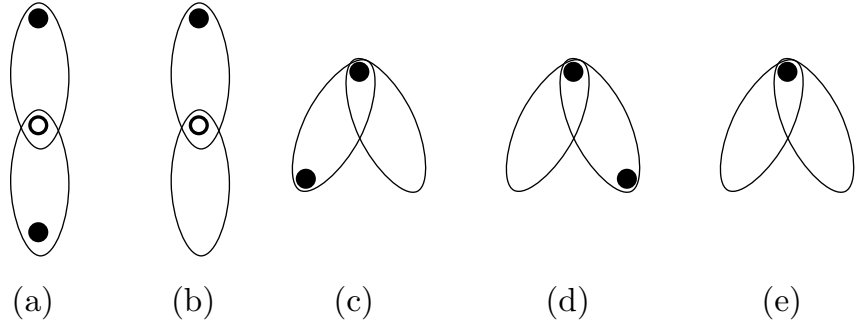
\FIGURE{h}{.7}{toptrees}{All 5 types of merges in a top tree decomposition. Full circles denote boundary nodes of the resulting cluster and empty circle denotes the boundary node of the merged clusters which is not boundary for the resulting cluster.}

\noindent
Let a relevant pair of clusters $(C_1,C_2)$ be a pair of clusters $C_1\in\TT_1$ and $C_2\in\TT_2$ that have at least one common leaf.
From Property~\ref{prop:log_height_of_toptree} and other properties of top tree decomposition, holds the following fact:
\begin{fact}\label{fact:toptree_properties}
 The following properties hold:
 \begin{enumerate}[label=(\roman*)]
  \item Every leaf in $T_i$ is in $\Oh(\log n)$ distinct clusters $C_i$.
  \item There are $\Oh(n\log^2n)$ relevant pairs of clusters.
  \item For every cluster with two boundary nodes, one of the boundary nodes is an ancestor of the other.
 \end{enumerate}
\end{fact}
\noindent
Recall that our aim is to count all shared stars in $T_1$ and $T_2$.
Our algorithm will process each relevant pair of clusters and count some particular stars for each such pair.
Now, for every shared star we define which relevant pair of clusters does it contribute to.

\paragraph{Representatives.}
Consider a non-leaf node $u\in T_i$.
We define $R_i(u)$, the representative cluster of $u$ in $T_i$, as the smallest cluster that contains $u$ in which $u$ is not a boundary node.
Note that $R_i(u)$ is the lowest common ancestor (in the top tree representing $T_i$) of all clusters that have $u$ as a boundary node.
Furthermore, this cluster must be formed by either (a) or (b)-type merge (all types are in Figure~\ref{fig:toptrees}) and the node $u$ is then the empty circle in the merge.
Consider a shared star $s$ on leaves $L=\{a,b,c,d\}$ and central nodes $c_1\in T_1$ and $c_2\in T_2$.
Let $R_1(c_1)$ and $R_2(c_2)$ be the representative clusters of $c_1$ and $c_2$.
We slightly abuse the notation and write $R_i(s):=R_i(c_i)$ identifying a star with its central node in the corresponding tree $T_i$.
As $s$ is a shared star then subtree of $T_i$ induced by the leaves from $L$ is a star in both trees $T_i$. 
Hence each of the clusters $R_i(s)$ contains at least $2$ leaves from $L$.
We say that the star $s$ is of type I if $R_1(s)$ and $R_2(s)$ have at least one common leaf with $L$, otherwise is of type II.
See Figure~\ref{fig:star_types} for an example.
We intentionally draw the clusters as if the trees were unrooted and do not specify the type of merge, because the relation between the clusters (up/down, left/right) is irrelevant in this case.
Notice that the configuration described in the bottom row of Figure~\ref{fig:star_types} is the only possible for a star of type II. 
More precisely, in a star $s$ of type II each of $R_i(s)$ contains exactly two leaves of $s$ and none of them is both in $R_1(s)$ and $R_2(s)$.

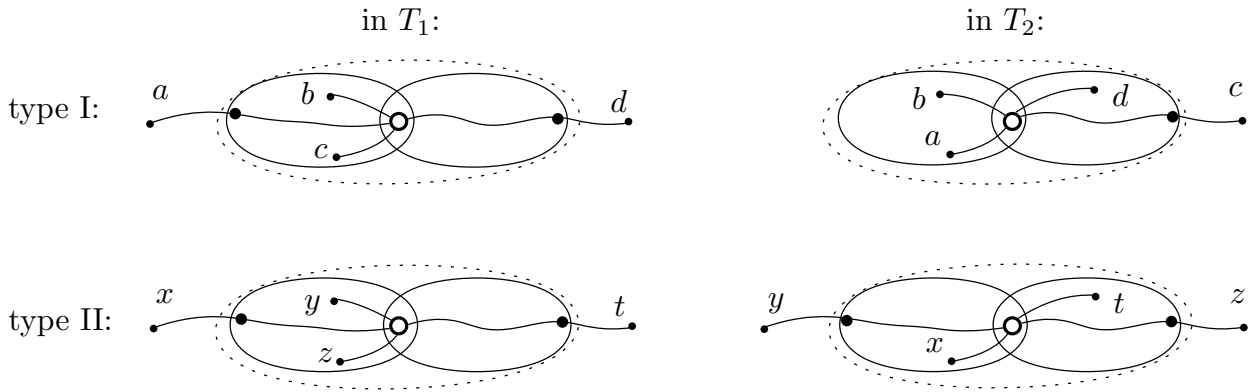
\FIGURE{h}{1}{star_types}{An example of a star of type I induced by leaves $a,b,c,d$ and a star of type II induced by leaves $x,y,z,t$.
For each star we mark the representative clusters with dotted lines and the central node with an empty circle.}

\begin{fact}
 If there exists an $i\in\{1,2\}$ such that $R_i(s)$ contains at least $3$ leaves of a shared star $s$, then the star $s$ is of type I.
\end{fact}

We define that stars $s$ of type I contribute to the pair of clusters $(R_1(s),R_2(s))$ and will be counted while considering this pair.
By definition, in this case $(R_1(s),R_2(s))$ is a relevant pair of clusters and hence will be considered explicitly by our algorithm.
However, for stars $s$ of type II, $(R_1(s),R_2(s))$ is not necessarily a relevant pair of clusters.
Let $R_i'(s)$ be the smallest cluster of $T_i$ containing at least 3 leaves of $s$.

\begin{lemma}\label{le:type_2_repr}
 For every star $s$ of type II, $R_i'(s)$ is uniquely defined and contains exactly 3 leaves common with $s$.
\end{lemma}
\begin{proof}
 Consider the cluster $R_i(s)$.
 As $s$ is of type II, two leaves of $s$ are outside of $R_i(s)$, at two opposite sides.
 Observe that every time we merge $R_i(s)$ with another cluster, we extend it from the side of only one boundary node.
 Hence we cannot simultaneously add both outside leaves to the cluster.
\end{proof}

We define that stars $s$ of type II contribute to the pair of clusters $(R_1'(s),R_2'(s))$. 
From Lemma~\ref{le:type_2_repr} it follows that $(R_1'(s),R_2'(s))$ is a relevant pair of clusters and hence will be considered explicitly by our algorithm.
Notice that, in this situation, the pair of clusters looks exactly as in Figure~\ref{fig:type_2_repr}.
To describe this in more detail, consider the subtree of $T_1$ induced by $s$ on the left side of Figure~\ref{fig:type_2_repr} and the names of nodes there.
One of the merged clusters (here left, $X_1$) contains the third, just added leaf ($x$) and the other (right, $Y_1$) has two leaves ($y,z$) connected to a single node on the path between two boundary nodes of the cluster.
In this situation, the type of merge (recall Figure~\ref{fig:toptrees}) can be arbitrary, not necessarily only (a) or (b) as for type I.

\begin{observation}
 For every star $s$ of type II, its central node is neither a boundary nor the merged boundary node of~$R'(s)$.
\end{observation}

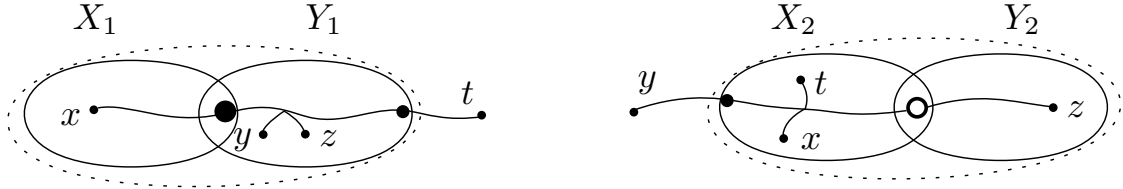
\FIGURE{h}{.9}{type_2_repr}{Clusters $R_1'(s)$ and $R_2'(s)$ for a star of type II connecting nodes $x,y,z,t$.
These are the smallest clusters containing three leaves from the star.
The central node of the star is neither a boundary nor the merged boundary node.
The type of merge can be arbitrary.}

\subsection{Counting Stars of Type II}\label{se:stars_type_II}

In this section we describe how to count stars of type II.
From Lemma~\ref{le:type_2_repr} it follows that it is enough to perform the calculations only for all relevant pairs of clusters and count stars of type II contributing to them.

First, we list all relevant pairs of clusters and store their common leaves.
This can be done by iterating over all leaves $\ell$ and then over all clusters containing $\ell$ in $\TT_1$ and then in $\TT_2$.
Now we consider every relevant pair of clusters $(C_1,C_2)$ and let $X_i$ and $Y_i$ be the clusters forming $C_i$.
We call $X_i$ and $Y_i$ subclusters to distinguish them from clusters $C_i$.
We group all stars of type II contributing to this pair by subclusters containing central nodes of the stars in $C_1$ (either $X_1$ or $Y_1$) and $C_2$ (either $X_2$ or $Y_2$).
We describe the calculations in detail for all stars induced by leaves $x,y,z,t$ with central nodes in clusters $Y_1$ and $X_2$, as in Figure~\ref{fig:type_2_repr}, the other cases are symmetric.

Let the leaves of a star $s$ be located in the clusters as in Figure~\ref{fig:type_2_repr}, that is in tree $T_1$ $x$ is in $X_1$, $y$~and~$z$ in $Y_1$ and $t$ is outside the considered cluster $C_1=R_1'(s)$, from the side of subcluster $Y_1$.
Regarding the location in $T_2$, among leaves $y$ and $z$, let $z$ be the leaf in $Y_2$ and $y$ outside $C_2=R_2'(s)$.
Lastly, $x$ and $t$ are in~$X_2$.
Let the spine of a cluster be the path connecting its two boundary nodes.
We say that a leaf $\ell$ connects to the spine $S$ in a node $u$ if $u$ is the closest node from $S$ to $\ell$.

We iterate over all leaves $z$ which are both in $Y_1$ and $Y_2$. 
From our assumptions, $z$ is the only common leaf of $Y_1,Y_2$ and $s$.
We need to count leaves $y$ such that:
\begin{enumerate}[label=(\roman*)]
 \item in $T_1$ connect to the same node on the spine of $Y_1$ as $z$, but with a different edge, and
 \item in $T_2$ are outside $C_2$, from the side of $X_2$.
\end{enumerate}
See Figure~\ref{fig:connect_to_spine} for the locations of $y$ with respect to $z$ in $Y_1$ in $T_1$ that we want to count or not.
Observe that the choice of leaves $x$ and $t$ is independent from $y$ and $z$.
Hence we can count pairs $x$ and $t$ in the same way as $y$ and $z$ and then multiply the obtained numbers.
Till the end of this subsection we focus only on counting pairs $y$ and $z$.

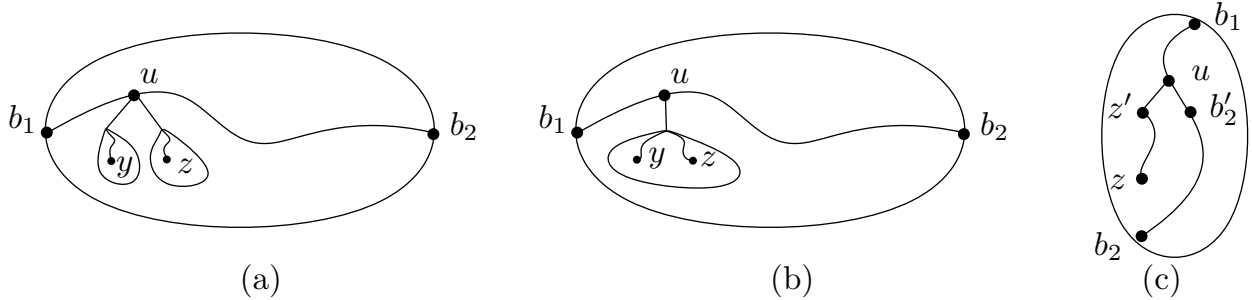
\FIGURE{h}{1.}{connect_to_spine}{For a fixed leaf $z$ we need to count leaves $y$ such that they connect to the spine in the same node as $z$, but with a different edge.
(a) Included, (b) excluded location of $y$ with respect to $z$. (c) For a fixed leaf $z$ we use nodes $u,z'$ and $b_2'$ to count all leaves $y$ satisfying both conditions (i) and (ii).}

Now we show that both the above conditions on $y$ can be phrased in terms of counting points in rectangles, for which we can use existing techniques.
\begin{lemma}\label{le:counting_leaves_y}
 After $\Oh(n\log n)$ time preprocessing of trees $T_1$ and $T_2$, for any leaf $z$ we can count leaves~$y$ which satisfy both conditions (i) and (ii) in $\Oh(\log n)$ time.
\end{lemma}
\begin{proof}
 Consider a pre-order numbering of all nodes of trees $T_i$.
 Every subtree of (unrooted) tree~$T_i$ corresponds either to one or two contiguous intervals of pre-order indices.
 Similarly, from the properties of top tree decomposition, the outside parts of a cluster also form one or two contiguous intervals of indices.
 Hence, a query about the number of common leaves in a subtree and the outside part of a cluster is the number of leaves with their pre-order indices in both trees inside particular ranges.
 This, in turn, can be answered efficiently using a constant number of queries about the number of points in rectilinear rectangles in the plane, also known as 2-D orthogonal range counting queries \cite{Chazelle88,JaJaMS04}, that can be answered in $\Oh(\log n)$ time
 after an $\Oh(n\log n)$ time preprocessing.
 
 However, condition (i) on $y$ is more involved than simply belonging to a particular subtree.
 Without loss of generality, suppose that $b_1$ is ancestor of $b_2$ (recall Fact~\ref{fact:toptree_properties}).
 Note that $z$ connects to the spine $b_1\cdots b_2$ in the lowest common ancestor (LCA) of $z$ and $b_2$, call this node $u$.
 Let $z'$ be the last node on the path from $z$ to $u$ and $b_2'$ be the last node on the path from $b_2$ to~$u$, as in Figure~\ref{fig:connect_to_spine}(c).
 After a linear-time preprocessing of $T_1$, node $u$, the LCA of $z$ and $b_2$, can be found in constant time \cite{SchieberV88}.
 With a slight modification (called the extended LCA query), we can also find nodes $z'$ and $b_2'$ in the same complexity\cite{GasieniecKPS05}.
 Now we need to count leaves $y$ that are in the subtree of $u$, but not in the subtree of $z'$ nor $b_2'$.
 To sum up, we described the condition (i) in terms of belonging or not to particular subtrees and hence the number of leaves $y$ satisfying both (i) and (ii) can be computed efficiently.
\end{proof}
\noindent
As every leaf belongs to $\Oh(\log n)$ clusters in each of the top trees, the following holds:

\begin{fact}\label{fact:sum_of_intersections}
 The total number of common leaves over all pairs of relevant clusters is $\Oh(n\log^2 n)$.
\end{fact}

To conclude, for every relevant pair of clusters we iterate over all their common leaves $z$ and count leaves $y$ satisfying both (i) and (ii).
By Fact~\ref{fact:sum_of_intersections} and Lemma~\ref{le:counting_leaves_y} counting all stars of type II takes $\Oh(n\log^2n\cdot \log n)=\Ohtilda(n)$.

\subsection{Counting Stars of Type I}

\newcommand{\BM}{\mathcal{M}}
\newcommand{\LL}{\mathcal{L}}

Recall that every star $s$ of type I contributes to the relevant pair of clusters $(R_1(s),R_2(s))$.
For this reason it is enough to iterate only over all relevant pairs of clusters and count stars contributing to the current pair.
From Fact~\ref{fact:toptree_properties} there are $\Oh(n\log^2 n)$ such pairs and from Fact~\ref{fact:sum_of_intersections} the overall number of common leaves in all of them is also $\Oh(n\log^2 n)$. 

We say that $(R_1(s),R_2(s))$ is the representative pair of star $s$.
Consider a relevant pair of clusters $(C_1,C_2)$.
Before we proceed with the general case, we first consider the following special case when a star $s$ with central nodes $c_1$ and $c_2$ has all its four leaves in clusters $C_1=R_1(s)$ and $C_2=R_2(s)$, or in other words, $s$ is fully contained in clusters of its representative pair.
Recall that in Section~\ref{se:from_qd_to_4c} we constructed a bipartite multigraph $\BM$ in such a way that nodes on the left (respectively right) correspond to subtrees attached to $c_1$ (respectively $c_2$).
See Figure~\ref{fig:star_to_m4} for an illustration.
We proceed similarly, except that we are only interested in counting stars with all leaves in both $C_1$ and $C_2$.
Therefore we redefine the multiplicity of an edge to be the number of such common leaves in the corresponding subtrees of $C_1$ and $C_2$.
We are interested only in edges with non-zero multiplicities, so some of the nodes might be isolated.
We would like to completely disregard such isolated nodes and construct the graph in time proportional to the number of edges with non-zero multiplicities.

Let $\LL$ be the set of common leaves of $C_1$ and $C_2$.
Notice that every such leaf contributes to one multi-edge of $\BM$.
We iterate over all leaves in $\LL$ and update the multiplicities of edges as follows.
Given a leaf, we extract the endpoints of its corresponding edge using extended LCA queries.
Then we look up the corresponding edge in a dictionary and, if it already exists, increase its multiplicity or create a new edge otherwise.
This allows us to construct $\BM$ in $\Oh(|\LL|\log|\LL|)$ time.
Clearly, it holds $|E(\BM)|\le |\LL|$.

It is crucial that the time of construction of $\BM$ depends only on $\LL$, because then Fact~\ref{fact:sum_of_intersections} implies that the overall time of constructing multigraphs $\BM$ for all pairs of relevant clusters is $\Ohtilda(n)$.

\paragraph{Complexity.}
Before we describe the algorithm for counting all stars, let us summarize the complexity of the approach presented so far.
Recall that $\pow$ is the smallest number such that there exists an algorithm counting 4-cycles in $m$-edge simple graphs in $\Oh(m^\pow)$ time.
Combining Theorem~\ref{thm:counting_multi} and Corollary~\ref{cor:c4_in_multigraphs} we obtain that we can count matchings of size 4 ($\sJ$) in multigraphs with $m$ edges in $\Oh(m^\pow\log^4m)=\Ohtilda(m^\pow)$ time.
The algorithm of Vassilevska Williams et al. \cite{WilliamsWWY15} runs in time $\Oh(m^{\frac{4\omega-1}{2\omega+1}})=\Oh(m^{2-\frac{3}{2\omega+1}})=\Oh(m^{1.478})$, as $\omega<2.373$ \cite{Gall14a,Williams12}, so $\pow<1.48$.

Let $m_i$ be the number of common leaves in the $i$-th considered relevant pair of clusters and hence also the bound on the number of edges in the $i$-th multigraph $\BM$.
From Fact~\ref{fact:sum_of_intersections} we know that $\sum_i m_i =\Oh(n\log^2n)$, where $i$ ranges over all relevant pairs of clusters.
So the overall time of counting all stars of type I is:

\begin{equation}\label{eq:complexity}
\sum_i \Ohtilda(m_i^\pow)=\Ohtilda\left(\sum_i m_i^\pow\right) =
\Ohtilda\left(\frac {\sum_i m_i}n n^\pow\right) = 
\Ohtilda\left(n^\pow\log^2n\right) = \Ohtilda(n^\pow) 
\end{equation}
 \noindent
where we used convexity of $x^\pow$ (as $\pow\geq 1$), $m_i\le n$ and $\sum_im_i=\Ohtilda(n)$.
To sum up, our algorithm counts all stars fully contained in their representative pairs in $\Ohtilda(n^\pow)=\Oh(n^{1.48})$ time.

\paragraph{Almost all stars of type I.}
\newcommand{\nNormal}{\ref{node_type:normal}\xspace}
\newcommand{\nMerged}{\ref{node_type:merged}\xspace}
\newcommand{\nOutside}{\ref{node_type:outside}\xspace}

Now we modify the above approach to count all stars of type~I, not necessarily fully contained in their representative pairs.
The main difficulty is that now the stars can contain leaves outside of $\LL$ and we cannot explicitly insert them as edges in the multigraph, as we want to keep the $\Ohtilda(|\LL|)$ running time.
We define a modified bipartite multigraph $\BM'$ in a similar way as before.
For every side of the bipartite graph $i\in\{1,2\}$, let a neighbor of $c_i$ be implicit if it does not contain a leaf from~$\LL$ nor contains an outside part of $C_i$, otherwise we call it explicit.
In $\BM'$ we have three types of nodes for every side $i$ of the bipartite graph:
\begin{enumerate}[label=(\arabic*)]
 \item\label{node_type:outside} at most two nodes for subtrees connected to $c_i$ that contain an outside part of the cluster~$C_i$,
 \item\label{node_type:normal} at most $|\LL|$ nodes for subtrees connected to $c_i$ that contain a leaf from $\LL$, but do not contain an outside part of $C_i$,
 \item\label{node_type:merged} one node representing all subtrees attached to the implicit neighbors of $c_i$.
\end{enumerate}
Thus, every node corresponds to a collection of subtrees of the whole (unrooted) $T_i$.
The multiplicity of an edge is simply the number of common leaves of subtrees corresponding to its endpoints (not necessarily only from $C_1$ and $C_2$).
See Figure~\ref{fig:3_types_in_multi}.

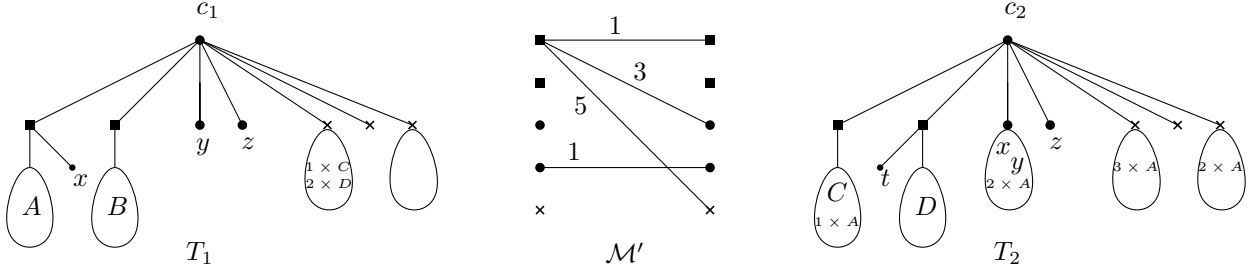
\FIGURE{h}{1}{3_types_in_multi}{Bipartite multigraph $\BM'$ with three types of nodes. The outside parts of the clusters are marked with capital letters.
We do not explicitly list leaves from the outside parts, but only their multiplicity (i.e. $3\times A$) which can be obtained with orthogonal queries.
Implicit neighbors of the central node are marked with crosses.
To avoid clutter, three edges with multiplicities $1$ are omitted.}

We need to show how to construct $\BM'$ in $\Oh(|\LL|\log|\LL|)$ time.
We start with listing all nodes of type \nNormal, similarly as we did for $\BM$.
We say that an edge is of type $(a)$-$(b)$ for $a,b\in\{1,2,3\}$, when it connects a node of type $(a)$ on the left side of the graph and $(b)$ on the other.
Now we describe how to construct in $\Oh(|\LL|\log |\LL|)$ time edges of each type separately.

\begin{enumerate}
 \item \nNormal-\nNormal: We obtain all these edges together with their multiplicities by iterating over all leaves from $\LL$, as we did for $\BM$.
 \item \nOutside-\nOutside: Even though we obtained some multiplicities of these edges while iterating over~$\LL$, we disregard them and use orthogonal queries for intersection of particular ranges to obtain the multiplicities.
 In total there are at most 2 nodes of type \nOutside at each side of the graph.
 \item \nOutside-\nNormal, \nNormal-\nOutside: Similarly as above, we disregard all edges of this type found while iterating over $\LL$ and use orthogonal queries to obtain the multiplicities.
 There are $\Oh(|\LL|)$ nodes of type \nNormal, so this step runs in $\Oh(|\LL|\log|\LL|)$ time.
 \item \nNormal-\nMerged, \nMerged-\nNormal, \nMerged-\nMerged: These edges always have multiplicity 0, because otherwise there would be a common leaf in the corresponding subtree making the nodes of different type.
 \item \nOutside-\nMerged, \nMerged-\nOutside:  We use orthogonal range queries and the multiplicities of edges computed so far to retrieve the multiplicities of the $\Oh(1)$ remaining edges.
\end{enumerate}

To conclude, we can construct the bipartite multigraph $\BM'$ with $\Oh(|\LL|)$ non-zero edges in $\Oh(|\LL|\log|\LL|)$ time.
Now we would like to count matchings of size 4 ($\sJ$) in $\BM'$, but this is not enough yet.

\paragraph{Missing stars.} 
We have not counted stars that have two leaves in subtrees attached to different implicit neighbors of $c_i$, because the subtrees are merged to one node of type \nMerged and, by counting matchings $\sJ$, we allow choosing at most one leaf from them.
Observe that stars that have two leaves in subtrees attached to implicit neighbors of $c_i$ in both clusters have no common leaves with $\LL$.
Hence they are of type II and are counted separately in Section~\ref{se:stars_type_II}.
To sum up, we only have not counted stars of type I with two leaves from the outside parts of exactly one of the clusters that are from subtrees of implicit neighbors of $c_i$.
We call such stars missing.
In terms of matchings, a missing star corresponds to choosing two edges from one node of type \nMerged that are incident to two different nodes of type~\nOutside.
Because we identified all nodes of type~\nMerged, this is not a matching in $\BM'$, but we need to count such stars as well.

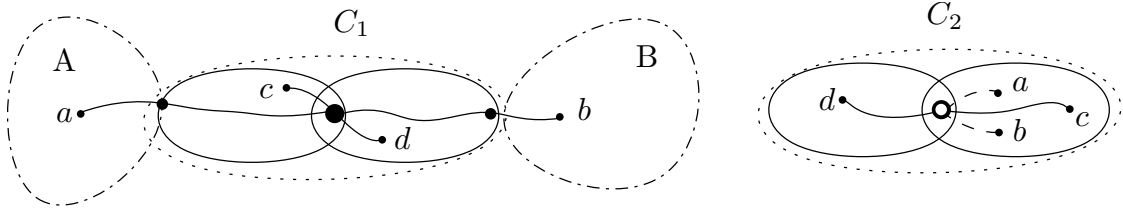
\FIGURE{h}{.9}{missing_star}{Missing star with leaves $a$ and $b$ from outside parts of $C_1$ and $c$ and $d$ that are both in $C_1$ and $C_2$.
Our pictures are rotated, that is the upper boundary of the cluster is on the left and the bottom one on the right.}

We show how to count all missing stars for which the situation described above takes place in cluster $C_2$, that is the missing star consists of two leaves $a$ and $b$ from the outside parts $A$ and $B$ of $C_1$.
Suppose that $B$ is the outside part below $C_1$ (in the rooting of $T_1$ that we use) and $A$ is above $C_1$.
See Figure~\ref{fig:missing_star}.
In $T_2$, leaves $a$ and $b$ are from two different subtrees attached to some implicit neighbors of $c_2$ which are merged together in $\BM'$.
Leaves $c$ and $d$ are from distinct explicit neighbors of $c_1$ and $c_2$.
Let $\alpha_i$ and $\beta_i$ be the number of leaves from $A$ and from $B$ in the $i$-th subtree connected to $c_2$.
Then the number of choices of leaves $a\in A$ and $b\in B$  that are in subtrees attached to two distinct implicit neighbors of $c_2$ is $\sum_{i\ne j} \alpha_i\beta_j$ where $i$ and $j$ range only over the implicit neighbors of $c_2$.
Now we need to multiply this number by the number of 2-matchings in $\BM'$ with three nodes deleted: two nodes of type~\nOutside from the side of $c_1$ and the node of type \nMerged from the side of $c_2$.
Using the notation from Appendix~\ref{se:missing_proofs} we can compute the number of 2-matchings in $\Oh(|\LL|)$ time:
$$\cc= = \frac12 \left(\sum_{(u,v)\in E}\mult(u,v)\multi{\EminusUV}{1}\right)$$

Now we need to compute $\sum_{i\ne j} \alpha_i\beta_j$ where $i$ and $j$ range only over implicit neighbors of $c_2$.
Notice that we can compute every single value of $\alpha_i$ or $\beta_j$ with orthogonal queries.
Next, we can omit the requirement that we iterate only through implicit neighbors, because we can compute the sum for all neighbors of $c_2$ and subtract appropriate terms for explicit nodes in $\Ohtilda(|\LL|)$ time.
Finally, $\sum_{i\ne j} \alpha_i\beta_j=(\sum_i\alpha_i)\cdot(\sum_i\beta_i)-\sum_i\alpha_i\beta_i=|A|\cdot|B|-\sum_i\alpha_i\beta_i$, so we can focus only on computing the last expression.
In the next paragraph we restate this subproblem again and describe in detail how to calculate the desired sum efficiently.

\paragraph{Computing $\sum_i\alpha_i\beta_i$.}

\newcommand{\Mark}{\textsf{Mark}}
\newcommand{\Count}{\textsf{Count}}

In the previous paragraph we distilled the following subproblem.
Consider a relevant pair of clusters $(C_1,C_2)$ with merged boundary node $c_2\in C_2$.
Let $A$ be the outside part of $T_1$ above $C_1$ (in the considered rooting of $T_1$) and $B$ below $C_1$.
All leaves of $T_2$ are marked with color $\perp,A$ or $B$ which denotes that the leaf is inside $C_1$, in the outside part $A$ or in $B$, respectively and we call such marking the marking with respect to $C_1$.
Our aim is to compute $\sum \alpha_i\beta_i$, where $\alpha_i$ and $\beta_i$ denote respectively the number of leaves of color $A$ and $B$, in the $i$-th subtree connected to $c_2$.
We need to count the sum for all relevant pairs of clusters efficiently.
To simplify the presentation, we proceed off-line, that is we will compute and store answers for all the above queries.

Our algorithm resembles the approach of Brodal et al. in Section 5 of \cite{BrodalFMPS13}.
We keep a separate data structure supporting the following operations on $T_2$ in $\Oh(\log n)$ time:

\begin{itemize}
 \item $\Mark(u,c)$ - marks node $u\in T_2$ with color $c\in\{A,B,\perp\}$,
 \item $\Count(u)$ - computes $\sum_i\alpha_i\beta_i$ where $i$ ranges over all neighbors of the node $u \in T_2$.
\end{itemize}
We consider all clusters of top tree $\TT_1$ in the order of DFS traversal starting at the root.
We maintain the following invariant during the traversal:
$$\textit{When entering a cluster }C\in\TT_1\textit{, all leaves in }T_2\textit{ are marked with respect to cluster }C.$$

\noindent
We start with the cluster representing the whole tree $T_1$ and all leaves in $T_2$ are marked with~$\perp$.
Then we traverse the top tree $\TT_1$ top-down and suppose that we consider cluster $C_1$ formed by merging clusters $C'$~and~$C''$.
From the invariant, all leaves in $T_2$ are appropriately marked and we can call $\Count(s_2)$ and store the result for all merged boundary nodes $s_2$ of clusters $C_2$ such that $(C_1,C_2)$ is a relevant pair.
Then, while entering cluster $C'$ we need to mark all leaves from $C''$ with color $A$ or $B$, depending on the location of $C''$, recurse and, while exiting, mark leaves from $C''$ with $\perp$.
Then we proceed similarly for $C''$.
From Fact~\ref{fact:toptree_properties} there will be $\Oh(n\log n)$ updates in total.

\begin{lemma}
 There exists a data structure supporting {$\normalfont\Mark(u,c)$} and {$\normalfont\Count(u)$} operations in $\Oh(\log n)$ time.
\end{lemma}
\begin{proof}
 Recall that $T_2$ is rooted, so we can apply heavy-light decomposition \cite{SleatorT83} to it.
 The root is called light and every node calls its child with the largest subtree (and the leftmost in case of ties) heavy and all other children light.
  
 For every node $v$ of $T_2$ we maintain the counter: $\sum_\ell\alpha_\ell\beta_\ell$ where $\ell$ ranges only over the light children of $v$ and counters $\alpha_\ell$ and $\beta_\ell$ for all light children of $v$.
 Every update (marking) of a node $w$ first changes the color of~$w$.
 Then we iterate over all its light ancestors and appropriately update their counters.
 From the properties of heavy-light decomposition, every node has $\Oh(\log n)$ light ancestors, so the update takes $\Oh(\log n)$ time.
 
 To answer the $\Count(u)$ query, we use the $\sum_\ell\alpha_\ell\beta_\ell$ counter for $u$ and need to add the values for its parent and the heavy child in the rooted tree, if they exist.
 We obtain the latter values using a constant number of orthogonal range queries in $\Oh(\log n)$ time.
\end{proof}
\noindent
To conclude, we can aggregately answer all queries of $\sum_i\alpha_i\beta_i$ in $\Ohtilda(n)$ time and then count all the missing stars.
Hence, we can count all shared stars of type I in $\Ohtilda(n^{\pow})$ time where $\Oh(m^\pow)$ is the best complexity of an algorithm counting 4-cycles in a graph with $m$ edges.
As we counted all stars of type II in $\Ohtilda(n)$ time, the whole algorithm counting shared stars in $T_1$ and $T_2$ runs in $\Ohtilda(n^{\pow})$ time.
This concludes the proof of Theorem~\ref{thm:quartet_final}.

\paragraph{Dependency on $d$.}
In this paragraph we analyze the complexity of the algorithm with respect to the maximum degree $d$ of an internal node.

Recall that our algorithm counts 4-cycles in multiple multigraphs. Let $n_i$ and $m_i$ be the number of nodes and edges
in the $i$-th considered multigraph.
When bounding the total complexity in \eqref{eq:complexity} we only used the fact that $m_i\le n$.
However, in our construction $n_i=\Oh(d)$, so also $m_i=\Oh(d^2)$.
Hence our algorithm runs in time $\Ohtilda(\frac{n}{d^2}d^{2\pow})=\Oh(nd^{0.96})$ as $\pow <1.48$ \cite{WilliamsWWY15}.

Notice that, for dense graphs it is more desirable to use the algorithm by Alon et al.~\cite{AlonYZ97} that runs in $\Oh(n^\omega)$
time where $\omega<2.373$ \cite{Gall14a,Williams12}.
This change decreases the complexity of our algorithm to $\Ohtilda\left(\sum_i \min(n_i^\omega,m_i^\pow)\right)$,
where $\sum_i n_i =\Ohtilda(n)$, $\sum_i m_i=\Ohtilda(n)$ and, for every $i$, it holds that $n_i=\Oh(d)$ and $m_i\le \min(n_i^2,n)$.

Bounding the sum $\sum_i \min(n_i^\omega,m_i^\pow)$ is not immediate, so we divide its terms into $\log^2n$ groups identified by
a pair $(k,\ell)$ of parameters such that $n_i\in (2^{k-1},2^k]$ and $m_i\in (2^{\ell-1},2^\ell]$.
Observe that there are at most $\Ohtilda(n/2^{\max(k,\ell)})$ terms in every group $(k,\ell)$ due to the bound on the total number of nodes and edges.
Now we divide all the groups into three categories, depending on the relation of $k$ and $\ell$:
$k<\frac\pow\omega\ell$ or $\frac\pow\omega\ell \le k < \ell$ or $\ell \le k$.
For each of them we bound the corresponding terms by~$\Ohtilda(nd^\delta)$ where $\delta=\omega-\omega/\pow<0.77$.
To sum up, the algorithm runs in $\Ohtilda(nd^{0.77})$ time.

\section{Acknowledgments}

We thank Yinzhan Xu for pointing out an error in our original proof in Section~\ref{se:small_mult}.

\bibliography{biblio}

\appendix

\section{Counting Shapes in Multigraphs}\label{se:missing_proofs}

\begin{proof}[Proof of Theorem~\ref{thm:counting_multi}]
 We will generalize the calculations from Lemmas \ref{le:m4}, \ref{le:calculations_easy} and \ref{le:calculations_hard} for multigraphs.
 Recall that every edge $e$ can appear in the graph multiple times, so we need to take $\mult(e)$ into account when counting all shapes containing an edge $e$.
 Informally, we need to be more careful while using the binomial coefficient.
 Let $\multi Sk$ denote the number of ways of choosing $k$ distinct edges from a set $S\subseteq E$ of multi-edges, that is we cannot take more than one copy of any edge.
 Recall that every edge appears in $E$ exactly once, but separately we also have a function $\mult$ that returns multiplicity of every edge $e\in E$.
 \begin{fact}
  $\multi Sk$ can be computed from $S$ in $\Oh(|S|k)$ time using dynamic programming.
 \end{fact}
 \noindent
 For our purposes, $k$ will be always at most $4$.
 We first memorize $\multi Sk $ for all $k\le 4$ and the following sets $S$: $E,\{e\}$ for all edges $e\in E$, and $E(v)$ (set of all edges incident to $v$) for all nodes~$v$.
 Then we can combine the memorized values to compute $\multi{C}{k}$ for different sets $C$ using the following lemma:
 \begin{lemma}
 Let $A$ and $B$ be sets of edges such that $B \subseteq A$ and $\multi Ai,\multi Bi$ are already computed for all $0\le i\le k$.
 Then all values $\multi{A\setminus B}{i}$, for $0\leq i\leq k$, can be computed in $\Oh(k^2)$ time.
 \end{lemma}
 \begin{proof}
 We compute $\multi{A\setminus B}{j}$ for $j=0,1,\ldots,k$ using the following property:
 $$\multi{A\setminus B}{j}=\multi Aj -\sum_{i=0}^{j-1}\multi{A\setminus B}{i}\multi {B}{j-i}\qedhere$$
 \end{proof}
 \noindent
 Now we consider the shapes as in Lemma~\ref{le:calculations_easy}.
 To simplify the notation, by $E-e$ we denote $E\setminus\{e\}$ and write $\mult(u,v)$ instead of $\mult(\{u,v\})$.
 \begin{enumerate}
 \item $\cc\sA=\sum_{u\in V_1}\multi{E(u)}{4}$
 \item $\cc\sB=\sum_{(u,v)\in E}\mult(u,v)\multi{E(u)-(u,v)}{2}\multi{E(v)-(u,v)}{1}$
 \item $\cc\sC=\left(\sum_{u\in V_1}\multi{E(u)}{3}\multi{E\setminus E(u)}{1}\right)-\cc\sB$
 \item Let $\cc>=\sum_{v\in V_2}\multi{E(v)}{2}$. 
 Now, instead of $\zz$, we need to count a shape similar to~$\zz$, but instead of choosing a single middle edge $e$ we select an ordered pair of edges $(e_1,e_2)$ where possibly $e_1=e_2$:
 
 $\cc{\zoz}=\sum_{(u,v)\in E}\left(\mult(u,v)\right)^2\multi{E(u)-(u,v)}{1}\multi{E(v)-(u,v)}{1}$.\\
 Then:
 $\cc\sG=\frac 12 \left(\sum_{(u,v)\in E}\mult(u,v)\multi{E(u)-(u,v)}{1}\left(\cc>-\multi{E(v)}{2}\right) -\cc{\zoz}-2\cc\sB-\cc{\mirror\sB}\right)$
 \end{enumerate}
 Now we consider the shapes as in Lemma~\ref{le:calculations_hard}. Again the $\ff{\text{-}}$values are auxiliary.
 \begin{enumerate}
  \item Let $\tc\sE =   \sumxy{\sum_{v\in V_2}}{\multi{E(x)-(x,v)}{1}}{\multi{E(y)-(y,v)}{1}}{N(v)}$.
  
  Then:
  $\cc\sE=\tc\sE -2\cc\sD=\tc\sE-2\xx$.
  
  \item Let $\ff\sF=\sumxynobreak{}{\multi{E(x)}{2}}{\multi{E(y)}{2}}{V_1}$.
  
  Then:
  $\cc\sF=\ff\sF - \cc\sE -\cc\sD=\ff\sF - (\tc{\sE}-2\xx) - \xx = \tc{\sF}+\xx$.
  
  \item Let $\ff\sH=\sum_{(u,v)\in E}\mult(u,v)\multi{E(u)-(u,v)}{1}\multi{E(v)-(u,v)}{1}\multi{\EminusUV}{1}$.
  
  Then:  
  $\cc\sH=\ff\sH-2\cc\sE -2\cc{\mirror{\sE}} - 4\cc\sD\\
  = \ff\sH -2(\tc{\sE}-2\xx)-2(\tc{\mirror\sE}-2\xx)-4\xx =\tc{\sH}+4\xx$.
  
  \item Let $\ff\sI=\sum_{(u,v)\in E}\mult(u,v)\multi{E(u)-(u,v)}{1}\multi{\EminusUV)}{2}$.
  
  Then:  
  $\cc\sI=\frac 12 \left(\ff\sI-2 \cc\sG - \cc{\mirror\sB}-\cc\sH-2\cc\sE-4\cc\sF \right)\\
  =\frac 12 \left( \ff\sI- 2\tc{\sG} - \tc{\mirror\sB} -(\tc{\sH}+4\xx) - 2(\tc{\sE}-2\xx) -4(\tc{\sF}+\xx) \right)\\
  =\tc{\sI}+\frac 12 \left(-4\xx+ 4\xx-4\xx \right)
  =\tc{\sI}-2\xx$.
\end{enumerate}
\noindent
Now we proceed to the main shape $\sJ$ as in Lemma~\ref{le:m4}.
\begin{enumerate}
  \item Let $\cc\zz=\sum_{(u,v)\in E}\mult(u,v)\multi{E(u)-(u,v)}{1}\multi{E(v)-(u,v)}{1}$ \\
  and $\cc\leq = \frac 12 \left(\sum_{(u,v)\in E}\mult(u,v)\multi{E(u)-(u,v)}{1}\multi{\EminusUV}{1} - \cc\zz\right)$. \\  
  Then: $\cc\equiv = \frac 13 \left(\sum_{(u,v)\in E}\mult(u,v)\multi{\EminusUV}{2} - \cc\leq -\cc\geq \right)$.
  \item 
  Similarly as in $\zoz$ we consider shapes $\sOVZ,\sOV$ and $\sOII$ in which an ordered pair of (not necessarily distinct) edges $(e_1,e_2)$ connects one pair of nodes.
  
  Let $\cc\sOVZ= \sum_{(u,v)\in E}\mult(u,v)\multi{E(u)-(u,v)}{1}\left(\multi{E(v)-(u,v)}{1}^2 -\multi{E(v)-(u,v)}{2}\right)$\\
  and $\cc\sOV = \sum_{(u,v)\in E}\left(\mult(u,v)\right)^2(s_1-\multi{E(u)}{2}) - \cc\sOVZ$ where $s_1=\sum_{u\in V_1}\multi{E(u)}{2}$\\
  and $\cc\sOII = \sum_{(u,v)\in E}\left(\mult(u,v)\right)^2\multi{\EminusUV}{2} - \cc\sOV -\cc{\mirror{\sOV}}$.
  
  Then: 
  $\cc\sJ = \frac 14 \left(m\cc\equiv -\cc\sH-2\cc\sI-2\cc{\mirror\sI}-\cc\sOII  \right)\\
 =\frac 14 \left(m\cdot\tc\equiv -(\tc{\sH}+4\xx)-2(\tc{\sI}-2\xx)-2(\tc{\mirror\sI}-2\xx) -\tc\sOII \right)\\
 =\tc{\sJ} +\frac 14 ( -4\xx+4\xx+4\xx)=\tc{\sJ}+\xx$.\qedhere
 \end{enumerate}
 
\end{proof}

\end{document}